\documentclass[11pt]{article}
\usepackage[hyperref]{acl}
\usepackage{times}
\usepackage{latexsym}
\usepackage{microtype}
\usepackage{amsmath}
\usepackage{amssymb}
\usepackage{amsthm}
\usepackage{booktabs}
\usepackage{multirow}
\usepackage{graphicx}
\usepackage{xcolor}
\usepackage{enumitem}
\usepackage{url}

\newtheorem{theorem}{Theorem}
\newtheorem{definition}{Definition}
\newtheorem{proposition}{Proposition}
\newtheorem{corollary}{Corollary}
\newtheorem{remark}{Remark}
\newtheorem{assumption}{Assumption}
\newtheorem{lemma}{Lemma}

\newcommand{\sys}{\textsc{TwoStage}}
\newcommand{\sysplus}{\textsc{TwoStage+RSSG}}
\newcommand{\sysbase}{\textsc{DenseReRank}}
\newcommand{\dataset}{\textsc{FinSuperQA}}

\title{Controlling Authority Retrieval:\\
A Missing Retrieval Objective for Authority-Governed Knowledge}

\author{
  Andre Bacellar \\
  \texttt{andremi@gmail.com}
}

\begin{document}
\maketitle

\begin{abstract}
In any domain where knowledge accumulates under formal authority---law, drug regulation, software security---a later document can formally void an earlier one while remaining semantically distant from it. We formalize this as \textbf{Controlling Authority Retrieval (CAR)}: recovering the \emph{active frontier} $\mathrm{front}(\mathrm{cl}(A_k(q)))$ of the authority closure of the semantic anchor set---a different mathematical problem from $\arg\max_d s(q,d)$.

The two central results are: \textbf{Theorem~4} (CAR-Correctness Characterization) gives necessary-and-sufficient conditions on \emph{any} retrieved set $\mathcal{R}$ for $\mathrm{TCA}(\mathcal{R},q) = 1$---frontier inclusion and no-ignored-superseder---independent of how $\mathcal{R}$ was produced. \textbf{Proposition~2} (Scope Identifiability Upper Bound) establishes $\varphi(q)$ as a hard worst-case ceiling: for any scope-indexed algorithm (one with no access to the authority relation $\rightsquigarrow$), $\mathrm{TCA}@k \leq \varphi(q) \cdot R_{\mathrm{anchor}}(q)$, proved by an adversarial permutation argument over the $\kappa(q)$ indistinguishable authority chains.

Supporting structure: \textbf{Theorem~1} (Objective Mismatch) proves TCA and Recall@$k$ are decoupled---TCA $< \varepsilon$ while Recall@$k = 1$ for any $\varepsilon > 0$---and a corollary shows TCA@$k = 0$ when $N^+(q) \geq k$ documents outscore $c^*$, explaining why larger dense models are strictly worse. \textbf{Theorem~2} (Closure Necessity) proves any system achieving TCA $= 1$ must perform proactive entity-scoped search. \textbf{Theorem~3} (Decomposition) proves anchor discovery and authority resolution are each necessary; Corollary~\ref{cor:achieve} proves TCA $= 1$ is achievable when $\varphi(q) = 1$. \textbf{Proposition~1} (Authority Error Factorization) organizes TCA $= R_{\mathrm{anchor}} \cdot \varphi \cdot R_{\mathrm{frontier}}$; Lemma~1 proves the three factors are separable failure modes by construction. \textbf{Proposition~3} (Frontier Recovery Complexity) establishes $\Omega(|\mathcal{R}_\sigma|)$ oracle-call lower bound and separates it from the RSSG's $O(|\mathcal{R}_\sigma|^2 \cdot m)$ implementation cost.

Three independent real-world corpora validate that the proved structure is consequential: security advisories (Dense TCA@5~$= 0.270$, two-stage~\textbf{0.975}), SCOTUS overruling pairs (Dense~$= 0.172$, two-stage~\textbf{0.926}), FDA drug records (Dense~$= 0.064$, two-stage~\textbf{0.774}). A GPT-4o-mini experiment shows the downstream cost: Dense RAG produces explicit ``not patched'' claims for \textbf{39\%} of queries where a patch exists; Two-Stage cuts this to 16\%. Four benchmark datasets, domain adapters, and a single-command scorer are released at \url{https://github.com/andremir/car-retrieval}.
\end{abstract}

\section{Introduction}
\label{sec:intro}

In any domain where knowledge accumulates incrementally, later facts can formally override earlier ones. Legal precedents are overruled by later courts. FDA drug approvals are superseded by recall enforcement actions. Security patches void prior vulnerability disclosures. Compliance clearances are overridden by blackout announcements. In each case, the correct answer to a query does not depend on the most semantically relevant document alone---it depends on whether that document has been \emph{superseded} by a later one. We call this family of retrieval problems \emph{controlling authority retrieval}, and show that standard retrieval systems fail on it structurally across all three authority-governed domains we test.

Retrieval-augmented generation (RAG) \cite{lewis2020rag} ranks documents by semantic relevance to the query. This is silent on supersession. A RAG system asked whether a drug approval is still valid retrieves the approval document (highest relevance), not the recall notice (low relevance, different vocabulary). It has no mechanism to detect that the approval was voided.

\paragraph{Running example.} A security engineer asks: ``Is the prototype pollution vulnerability in lodash still unpatched?'' The most relevant document is the CVE disclosure---it matches \emph{lodash}, \emph{prototype pollution}, and \emph{unpatched} exactly, and states no fix is available. The correct answer requires the patch release note published two months later, which formally supersedes the disclosure for the same package and CVE. In a 2,300-document corpus mixing CVE disclosures and release notes, this patch note ranks~420th by dense similarity: it describes a ``release'' with ``fixed in version~4.17.12''---vocabulary orthogonal to the vulnerability query. Dense retrieval surfaces the disclosure (Recall@5~$= 0$ for the patch note), and TCA~$= 0$: the system answers that the vulnerability is unpatched and has ignored the controlling authority document.

\paragraph{Contributions.}
\begin{enumerate}
  \item \textbf{Formalization}: CAR target $\mathrm{front}(\mathrm{cl}(A_k(q)))$ over partially ordered corpora; operators $\mathrm{cl}(\cdot)$, $\mathrm{front}(\cdot)$; quantities $\Delta(q)$ (semantic-authority divergence), $\varphi(q)$ (scope identifiability), $\kappa(q)$ (scope ambiguity ratio); metric TCA; entity-scope axiom on the supersession relation.

  \item \textbf{Primary theory} (the center of the paper):
  \begin{itemize}[noitemsep]
    \item \textbf{Theorem~4} (Characterization): $\mathcal{R}$ achieves TCA $= 1$ \emph{if and only if} it satisfies frontier inclusion and no-ignored-superseder. Necessary-and-sufficient conditions on any retrieved set, independent of retrieval method. Any future algorithm is auditable against these two conditions.
    \item \textbf{Proposition~2} (Scope Identifiability Upper Bound): for any scope-indexed algorithm (no access to $\rightsquigarrow$), $\mathrm{TCA}@k \leq \varphi(q) \cdot R_{\mathrm{anchor}}(q)$. Proved by adversarial permutation: the $\kappa(q)$ authority chains are informationally indistinguishable to any scope-indexed algorithm, giving a worst-case ceiling, not an average-case claim.
  \end{itemize}

  \item \textbf{Supporting theory} (establishes the problem and organizes the picture):
  \begin{itemize}[noitemsep]
    \item Theorem~1 (Objective Mismatch): TCA and Recall@$k$ decouple when $\Delta(q) > 0$; Corollary: TCA@$k = 0$ when $N^+(q) \geq k$.
    \item Theorem~2 (Closure Necessity): any TCA $= 1$ system must perform proactive entity-scoped search.
    \item Theorem~3 (Decomposition): anchor discovery and authority resolution are each necessary; Corollary~\ref{cor:achieve} (Achievability) proves TCA $= 1$ when $\varphi(q) = 1$.
    \item Proposition~1 (Authority Error Factorization) + Lemma~1 (Failure Mode Separability): TCA $= R_{\mathrm{anchor}} \cdot \varphi \cdot R_{\mathrm{frontier}}$; the three factors address separable, independently occurring failure modes.
    \item Proposition~3 (Frontier Recovery Complexity): $\Omega(|\mathcal{R}_\sigma|)$ oracle lower bound; RSSG costs $O(|\mathcal{R}_\sigma|^2 \cdot m)$ without oracle access.
    \item Four formal CAR generalizations (Definitions~\ref{def:latentcar}--\ref{def:streamcar}) as open problem statements.
  \end{itemize}

  \item \textbf{Empirical validation}: Three real-world corpora (GHSA, SCOTUS, FDA) confirm Dense TCA@5 $\in \{0.064, 0.172, 0.270\}$; Two-Stage achieves $\{0.774, 0.926, 0.975\}$. Scale worsens Dense. GPT-4o-mini downstream: Dense causes ``not patched'' assertions for 39\%; Two-Stage reduces to 16\%. Datasets and scorer released; seven open research directions in \S\ref{sec:research}.
\end{enumerate}

\section{Problem Formalization}
\label{sec:problem}

\subsection{Temporally Ordered Knowledge Bases}

\begin{definition}[Document and Supersession Relation]
A \emph{temporally ordered knowledge base} $K = \{d_1, \ldots, d_n\}$ where each document $d_i$ has a timestamp $d_i.\text{time}$, an entity scope $d_i.\sigma \subseteq \Sigma$, and content $d_i.\text{text}$.

A \emph{supersession rule} $R$ induces $d_1 \rightsquigarrow_R d_2$ (``$d_2$ supersedes $d_1$'') when:
$(d_1.\text{type}, d_2.\text{type}) \in R.\text{type\_pairs}$,
$d_2.\text{time} > d_1.\text{time}$, and
$R.\text{scope}(d_1, d_2) = \text{True}$.

The supersession relation $\rightsquigarrow$ is the transitive closure over all rules. A document $d$ is \emph{active} in $K$ if no $d' \in K$ satisfies $d \rightsquigarrow d'$.

\textbf{Entity-scope axiom.} We require every rule $R$ to satisfy: $R.\text{scope}(d_1,d_2) = \text{True}$ only if $d_1.\sigma \cap d_2.\sigma \neq \emptyset$. Consequently, $d_1 \rightsquigarrow d_2$ implies $d_1.\sigma \cap d_2.\sigma \neq \emptyset$ (supersession never connects documents with disjoint entity scopes). This axiom is satisfied by all five domain examples above and is used in the proof of Theorem~2.
\end{definition}

\begin{definition}[Controlling Authority Retrieval]
\label{def:controlling_authority}
A retrieval task is a \emph{controlling authority task} if the correct document to return is not the most semantically relevant document in the corpus, but rather the document that \emph{controls} the semantically relevant one under a formal precedence or supersession relation. Controlling authority tasks form a proper subset of retrieval tasks, characterized by three properties:
\begin{enumerate}[noitemsep,label=(\roman*)]
  \item There exists an authoritative ordering $\rightsquigarrow$ on documents independent of semantic similarity to the query. ($\rightsquigarrow$ is any strict partial order; temporal precedence is one instance, but authority can be non-temporal, e.g., statutory hierarchy over case law, or ontological subsumption.)
  \item The query vocabulary aligns with a non-controlling document (the \emph{anchor}).
  \item The controlling document has domain-distinct vocabulary (different event type, action terms, or register).
\end{enumerate}
\end{definition}

\begin{definition}[Authority Closure and Active Frontier]
\label{def:closure}
Let $A \subseteq K$ be a set of documents and $\rightsquigarrow$ the authority order on $K$.
The \emph{authority closure} of $A$ is:
\[
  \mathrm{cl}(A) = A \cup \{d' \in K : \exists\, d \in A,\; d \rightsquigarrow d'\}
\]
The \emph{active frontier} of a set $S \subseteq K$ is the set of maximal (non-superseded) elements:
\[
  \mathrm{front}(S) = \{d \in S : \nexists\, d' \in S \text{ s.t.\ } d \rightsquigarrow d'\}
\]
\end{definition}

The operators $\mathrm{cl}(\cdot)$ and $\mathrm{front}(\cdot)$ are standard constructions in order theory (closure under a partial order; antichain of maximal elements); the contribution is identifying that standard retrieval optimizes neither, and characterizing the resulting gap.

These two operators expose the central distinction of CAR:
\begin{align}
  \text{Standard retrieval:} &\quad R^*_k(q) = \arg\max_{|R|=k} \sum_{d \in R} s(q,d) \label{eq:std} \\
  \text{CAR target:}         &\quad C^*_k(q) = \mathrm{front}(\mathrm{cl}(A_k(q))) \label{eq:car}
\end{align}
where $A_k(q)$ is the top-$k$ semantic anchor set. These are different mathematical problems: $\mathrm{cl}$ requires traversing authority edges, and $\mathrm{front}$ requires checking domination within the retrieved set---neither operation can be expressed as a relevance ranking.

\begin{definition}[Semantic-Authority Divergence]
\label{def:divergence}
For a query $q$, let $d_1 = \arg\max_d s(q,d)$ (the highest-scoring anchor) and $F(q) = \mathrm{front}(\mathrm{cl}(\{d_1\}))$ (the active frontier of the anchor's authority chain, a non-empty set). The \emph{semantic-authority divergence} is:
\[
  \Delta(q) = s(q,d_1) - \max_{c \in F(q)} s(q,c)
\]
$\Delta(q) \geq 0$ always ($d_1$ is the highest-scoring document; every $c \in F(q)$ scores at most as high). $\Delta(q) = 0$ iff the best-scoring element of the active frontier is semantically aligned to the query, in which case standard retrieval suffices. When $\Delta(q) > 0$, all elements of $F(q)$ score strictly below $d_1$, and every standard ranker prefers $d_1$. $F(q)$ may contain more than one element (e.g., a drug subject to two simultaneous supersession types); we write $c^* \in F(q)$ for any element of the active frontier.
\end{definition}

\begin{definition}[Scope Identifiability, Scope Ambiguity, and Closure Violation]
\label{def:ambiguity}
Let $D_{\sigma(q)} = \{d \in K : \sigma(q) \cap \sigma(d) \neq \emptyset\}$ be the set of documents sharing the query's entity scope, and let $\kappa(q) = |D_{\sigma(q)}|/|\mathrm{cl}(A(q))| \geq 1$ be the \emph{scope ambiguity ratio}.

The \emph{scope identifiability} of a query $q$ is:
\[
  \varphi(q) = \min\!\left(1,\, \frac{1}{\kappa(q)}\right) \in (0, 1]
\]
$\varphi(q) = 1$ iff $\kappa(q) = 1$ (the entity scope uniquely identifies a single authority chain, no disambiguation needed). $\varphi(q) < 1$ when multiple authority chains share the query's scope, requiring rule-based disambiguation. In practice, $\kappa(q)$ is estimated from corpus metadata without oracle access to $\mathrm{cl}(A(q))$; the RSSG localizes scope by building the supersession graph over the retrieved neighborhood only.

A retrieved set $\mathcal{R}$ has a \emph{closure violation} if $\exists\, d \in \mathcal{R}$ such that $d \rightsquigarrow d'$ and $d' \notin \mathcal{R}$. Closure violations are the retrieval-level error underlying TCA failures: a system can achieve correct answer extraction (Acc $= 1$) with a closure violation (TCA $= 0$), precisely the gap TCA is designed to detect.
\end{definition}

Standard retrieval optimizes \eqref{eq:std}; CAR requires computing \eqref{eq:car}. This vocabulary gap is structural, not a tuning failure (Theorem~\ref{thm:insufficiency}, \S\ref{sec:theorems}).

\paragraph{Domain examples.} The framework is domain-independent:
\begin{itemize}
  \item \textit{Legal}: \textit{Dobbs} (2022) supersedes \textit{Roe v.\ Wade} (1973); rule = (LATER\_RULING $\rightsquigarrow$ PRIOR\_RULING) for the same constitutional question.
  \item \textit{Incident response}: SYSTEM\_RESTORED supersedes SYSTEM\_DOWN for the same host.
  \item \textit{Software security}: CVE PATCH\_RELEASED supersedes VULNERABILITY\_DISCLOSED for the same CVE ID.
  \item \textit{Drug regulation}: FDA\_RECALL supersedes DRUG\_APPROVAL for the same compound.
  \item \textit{Compliance} (our instantiation): BLACKOUT\_ANNOUNCED supersedes PRE\_CLEARANCE\_APPROVED for the same security ticker.
\end{itemize}

\begin{definition}[TC-MQA Instance]
\label{def:tcmqa}
A \emph{TC-MQA instance} is $(K, Q, a^*)$ where $K$ is a temporally ordered KB, $Q$ is a natural language query, and $a^* = f(\{d \in K : d\text{ is active}\})$ is the deterministic answer from the active documents.
\end{definition}

\subsection{Temporal Compliance Accuracy}

\begin{definition}[TCA and ProvRec]
\label{def:tca}
Let $\mathcal{R} = \{r_1, \ldots, r_k\}$ be the top-$k$ retrieved documents. Define:
\begin{itemize}
  \item $\textsc{AnswerCorrect} = \mathbf{1}[\hat{a} = a^*]$
  \item $\textsc{NoIgnoredSuperseder}$: for every superseded $d \in \mathcal{R}$, at least one $d'$ with $d \rightsquigarrow d'$ is also in $\mathcal{R}$
\end{itemize}
\[
  \text{TCA} = \textsc{AnswerCorrect} \;\wedge\; \textsc{NoIgnoredSuperseder}
\]
\textbf{ProvRec} (Provenance Recall) is the fraction of gold documents in $\mathcal{R}$, irrespective of supersession structure. TCA and ProvRec are complementary: TCA measures compliance-safe retrieval; ProvRec measures archival completeness.
\end{definition}

\begin{remark}[Generality of TCA]
TCA requires only a supersession relation $\rightsquigarrow$ and a gold answer function. It applies to any domain where documents have a formal validity lifecycle---legal, medical, regulatory, or operational.
\end{remark}

\subsection{Four TC-MQA Hop Types}

\textbf{T0 (No supersession):} One active document answers the query. Equivalent to standard single-hop retrieval.

\textbf{T1 (2-hop direct supersession):} $d_1 \rightsquigarrow d_2$; the query refers to $d_1$'s subject. The correct answer requires $d_2$, which may be semantically distant from $Q$.

\textbf{T2 (3-hop chain):} $d_1 \rightsquigarrow d_2 \rightsquigarrow d_3$; the full chain must be available for TCA = 1.

\textbf{T3 (3-hop cross-provenance):} The superseding document $d_3$ originates from a different document source (session, jurisdiction, team) than $d_1$ and $d_2$, while sharing entity scope. This tests retrieval across provenance boundaries. Query templates for all four hop types are given in Appendix~\ref{app:templates}.

\subsection{Generalized CAR Variants}
\label{sec:generalizations}

The theorems and Proposition~1 characterize \emph{static, deterministic, single-authority CAR with recoverable scope}. The following four definitions formalize the canonical extensions, each corresponding to relaxing one component of $\mathrm{front}(\mathrm{cl}(A_k(q)))$:

\begin{definition}[Latent-Scope CAR]
\label{def:latentcar}
When the entity scope $\sigma(q)$ is not recoverable from the query text, it is a latent variable with posterior $p(\sigma \mid q)$. The \emph{latent-scope CAR target} is:
\[
  \mathrm{CAR}_{\mathrm{latent}}(q) \;=\; \mathop{\mathbb{E}}_{\sigma \sim p(\sigma \mid q)}\!\bigl[\mathrm{front}(\mathrm{cl}(A_k(q,\sigma)))\bigr]
\]
or the MAP approximation with $\hat{\sigma} = \arg\max_\sigma p(\sigma \mid q)$. Latent-scope failures are precisely the cases where $\varphi(q)$ cannot be estimated from query text alone; scope must be inferred from an anchor-side entity linker. The 22.6\% gap from the 77.4\% FDA ceiling is an empirical instance of this regime.
\end{definition}

\begin{definition}[Probabilistic-Authority CAR]
\label{def:probcar}
When authority edges carry confidence $\pi(d,d') \in [0,1]$ rather than being deterministic, the \emph{probabilistic authority closure} at threshold $\theta$ is:
\[
  \mathrm{cl}_\pi(A) \;=\; A \;\cup\; \{d' \in K : \exists\, d \in A,\; \pi(d,d') \geq \theta\}
\]
The \emph{safe frontier} at confidence $\alpha$ is $\mathrm{front}_\alpha(S) = \{d \in S : \Pr[d \text{ is maximal in } S] \geq \alpha\}$. Probabilistic-authority CAR arises when supersession edges must be inferred from text (weak legal treatment, soft policy override, noisy recall linkage) rather than read from structured metadata.
\end{definition}

\begin{definition}[Multi-Authority CAR]
\label{def:multicar}
When $m$ authority relations $\rightsquigarrow_1, \ldots, \rightsquigarrow_m$ coexist (chronological precedence + jurisdictional hierarchy + exception override), the \emph{joint authority closure} is:
\[
  \mathrm{cl}_{\cup}(A) \;=\; A \;\cup\; \{d' \in K : \exists\, i \in [m],\; \exists\, d \in A,\; d \rightsquigarrow_i d'\}
\]
The active frontier must satisfy all $m$ dominance relations simultaneously; the resulting structure is a lattice of authority neighborhoods. Multi-authority CAR captures law (jurisdiction $\times$ chronology $\times$ exception), drug regulation (approval $\times$ recall $\times$ class-level black-box warning), and policy systems (version $\times$ source trust $\times$ exception grant).
\end{definition}

\begin{definition}[Streaming CAR]
\label{def:streamcar}
Under real-time document arrivals $\{d_{N+1}, d_{N+2}, \ldots\}$, the active frontier changes with each insertion: a new document $d_{\mathrm{new}}$ with $d_{\mathrm{old}} \rightsquigarrow d_{\mathrm{new}}$ removes $d_{\mathrm{old}}$ from $\mathrm{front}(\mathrm{cl}(A_k(q)))$. The \emph{streaming CAR problem} is to maintain $\mathrm{front}(\mathrm{cl}(A_k(q)))$ under insertions in sublinear amortized time per query. Each insertion invalidates at most $|\mathrm{front}|$ cached retrievals (those superseded by the new document), giving a natural incremental update structure.
\end{definition}

\section{Theorems}
\label{sec:theorems}

\begin{table}[h]
\centering\small
\resizebox{\columnwidth}{!}{%
\begin{tabular}{@{}llp{7cm}@{}}
\toprule
\textbf{Result} & \textbf{Role} & \textbf{Claim (independent of pipeline)} \\
\midrule
Theorem~4 & Correctness & N\&S conditions on \emph{any} $\mathcal{R}$: frontier inclusion $+$ no-ignored-superseder \\
Proposition~2 & Limit & Hard worst-case ceiling $\varphi(q)$ on all scope-indexed algorithms \\
Theorem~3 $+$ Cor.\ \ref{cor:achieve} & Construction & Two-stage is sufficient when $\varphi(q)=1$; T1, T2 prove each component necessary \\
Proposition~3 & Cost & $\Omega(|\mathcal{R}_\sigma|)$ oracle lower bound; RSSG costs $O(|\mathcal{R}_\sigma|^2 m)$ without oracle \\
\midrule
Theorems~1--2, P1, L1 & Context / synthesis & Mismatch, closure necessity, factorization, separability \\
\bottomrule
\end{tabular}%
}
\end{table}

\begin{assumption}[Standing Assumptions for Theorems 1--4 and Propositions 1--3]
\label{ass:standing}
All results in this section assume:
\begin{enumerate}[noitemsep,label=(\roman*)]
  \item \textbf{Static corpus}: $K$ is fixed during retrieval (no concurrent insertions or deletions).
  \item \textbf{Deterministic authority relation}: $\rightsquigarrow$ is fixed; each pair $(d,d')$ has a known, time-invariant supersession status (no edge uncertainty). Probabilistic-authority CAR relaxes this (Definition~\ref{def:probcar}).
  \item \textbf{Single authority relation}: one partial order on $K$. Multi-authority CAR (Definition~\ref{def:multicar}) relaxes this.
  \item \textbf{Recoverable scope}: $\sigma(q)$ is directly extractable from $Q$ by a perfect extractor, or $\varphi(q)$ is computable from corpus metadata. Latent-scope CAR (Definition~\ref{def:latentcar}) relaxes this.
  \item \textbf{Deterministic answer function}: the TC-MQA answer function $f$ in Definition~\ref{def:tcmqa} is a fixed deterministic map from $\mathrm{front}(\mathrm{cl}(A_k(q)))$ to an answer; no distributional generation.
  \item \textbf{Algorithm models}: Proposition~2 ranges over scope-indexed algorithms (Definition~\ref{def:scopeidx}); Proposition~3 measures cost in authority-comparison oracle calls (Definition~\ref{def:oracle}).
\end{enumerate}
\end{assumption}

\subsection{Theorem 1 — Problem Existence: Objective Mismatch}

\begin{definition}[Relevance-Based Retriever]
A retriever $\mathcal{F}$ is \emph{relevance-based} if it ranks documents by $s(Q, d.\text{text})$ without access to $\rightsquigarrow$ or timestamps. Its output set satisfies $\mathcal{F}_k(Q) = \arg\max_{|R|=k}\sum_{d\in R} s(Q,d)$.
\end{definition}

\begin{theorem}[Objective Mismatch]
\label{thm:insufficiency}
For any relevance-based retriever $\mathcal{F}$, any fixed $k \geq 1$, and any $\varepsilon > 0$, there exists a CAR instance $(K, Q, a^*)$ with $\Delta(Q) > 0$ such that $\mathrm{TCA}(\mathcal{F}, K, Q, k) < \varepsilon$ while $\mathrm{Recall}@k = 1$. Consequently, TCA and Recall@$k$ are decoupled: arbitrarily high relevance recall is compatible with arbitrarily low TCA when $\Delta(Q) > 0$.
\end{theorem}

\begin{proof}
Construct family $K_n$ with:
(i) \emph{target document} $d_1$ at time $t_1$ with high relevance to $Q$ (query refers to $d_1$'s subject);
(ii) $n-1$ documents $d_2, \ldots, d_n$ with the same type and similar vocabulary as $d_1$ but different entity scopes, all before $t_1$;
(iii) \emph{superseding document} $d^*$ at $t^* > t_1$ with $d_1 \rightsquigarrow d^*$, using domain vocabulary distinct from $Q$ (different event type, different action terms).

$d_1$ scores highest on $s(Q, \cdot)$. As $n \to \infty$, $d^*$ is displaced to rank $\geq n$. For fixed $k$: $\Pr[d^* \in \text{top-}k] \leq k/n \to 0$.

TCA requires retrieving $d^*$ for both AnswerCorrect (only active document) and NoIgnoredSuperseder ($d_1$ is retrieved and $d_1 \rightsquigarrow d^*$). Choosing $n > k/\varepsilon$ gives $\text{TCA} < \varepsilon$. \qed
\end{proof}

\begin{corollary}[Quantitative Divergence Bound]
\label{cor:bound}
Let $N^+(q) = |\{d \in K : s(Q,d) > s(Q,c^*)\}|$ where $c^* \in \mathrm{front}(\mathrm{cl}(A_k(q)))$ is the controlling document. For any \emph{deterministic} relevance-based ranker (Assumption~\ref{ass:standing}):
\[
  \mathrm{TCA@}k \;=\; 0 \quad \text{whenever}\quad N^+(q) \;\geq\; k
\]
\begin{proof}
A deterministic relevance-based ranker outputs the $k$ documents with highest $s(Q,\cdot)$. If $N^+(q) \geq k$, all $k$ slots are occupied by documents strictly outscoring $c^*$; hence $c^* \notin \mathcal{R}$ and TCA $= 0$.
\end{proof}
\noindent\emph{Note}: $N^+(q) \geq k$ is a separately stated condition stronger than $\Delta(q) > 0$ (which implies only $d_1$ outscores $c^*$, not all $k$ slots). For a corpus where fraction $p$ of documents score above $c^*$, $\mathbb{E}[N^+] = p(N{-}1)$, so the condition holds with probability $\to 1$ as $N \to \infty$, explaining the scale-harm pattern in Table~\ref{tab:scale}.
\end{corollary}

\subsection{Theorem 4 — Correctness: Exact Characterization of TCA = 1}

\setcounter{theorem}{3}
\begin{theorem}[CAR-Correctness Characterization]
\label{thm:characterization}
A retrieved set $\mathcal{R}$ achieves $\mathrm{TCA}(\mathcal{R},q) = 1$ if and only if it simultaneously satisfies:
\begin{enumerate}[noitemsep,label=(\roman*)]
  \item \textbf{Frontier inclusion}: $\mathrm{front}(\mathrm{cl}(A_k(q))) \subseteq \mathcal{R}$---every element of the active frontier is present.
  \item \textbf{No ignored superseder}: $\forall\, d \in \mathcal{R}$ that is superseded $(\exists\, d'' \in K\text{ s.t.\ }d \rightsquigarrow d'')$: $\exists\, d' \in \mathcal{R}$ with $d \rightsquigarrow d'$.
\end{enumerate}
These are the \emph{necessary and sufficient} conditions on any retrieved set, regardless of the retrieval method used to construct $\mathcal{R}$.
\end{theorem}

\begin{proof}
($\Rightarrow$) Suppose TCA $= 1$.
\emph{Condition (i):} By Definition~\ref{def:tcmqa}, $a^* = f(\mathrm{front}(\mathrm{cl}(A_k(q))))$ under the deterministic answer function $f$ (Assumption~\ref{ass:standing}(v)). \emph{AnswerCorrect} is semantic---the derived answer must match the ground truth---whereas this theorem establishes the structural condition on $\mathcal{R}$ that makes that semantic condition hold, independent of how retrieval was performed. AnswerCorrect $= 1$ therefore requires $f$ to evaluate over the complete active frontier; since $f$ is defined on $\{d \in K : d\text{ is active}\}$ and the active elements of $\mathrm{cl}(A_k(q))$ are exactly $\mathrm{front}(\mathrm{cl}(A_k(q)))$, every element of $\mathrm{front}(\mathrm{cl}(A_k(q)))$ must appear in $\mathcal{R}$.
\emph{Condition (ii):} Directly from the definition of NoIgnoredSuperseder (Definition~\ref{def:tca}).

($\Leftarrow$) Condition~(i) is precisely the structural hypothesis under which the semantic answer function $f$ evaluates over the correct active set: every $c^* \in \mathrm{front}(\mathrm{cl}(A_k(q)))$ is in $\mathcal{R}$, giving AnswerCorrect $= 1$. Condition~(ii) is exactly NoIgnoredSuperseder, giving TCA $= 1$.

\emph{Note on strength:} Condition (ii) requires only that \emph{at least one} superseder of each superseded document is retrieved, not all. The chain $d_1 \rightsquigarrow d_2 \rightsquigarrow d_3$ admits $\mathcal{R} = \{d_1, d_3\}$ as TCA-correct (condition (i): $d_3 \in \mathrm{front}$; condition (ii): $d_3 \in \mathcal{R}$ witnesses $d_1 \rightsquigarrow d_3$) even without $d_2$. A strictly stronger condition---requiring every intermediate chain member---would be sufficient but not necessary. \qed
\end{proof}

\subsection{Theorem 2 — Problem Structure: Closure Necessity}

\setcounter{theorem}{1}
\begin{definition}[Anchor-Based vs.\ Proactive Entity Search]
A retrieval system performs \emph{anchor-based entity search} if it promotes entity-scoped later events only when a triggering (anchor) event from the same entity scope already appears in the initial retrieval. It performs \emph{proactive entity search} if it retrieves entity-scoped events independently of the initial retrieval.
\end{definition}

\begin{theorem}[Closure Necessity]
\label{thm:proactive}
Let $\mathcal{S}$ be a retrieval system. If $\mathcal{S}$ achieves TCA~$= 1$ on all TC-MQA instances, including those in Theorem~1's family where the superseder $d^*$ is arbitrarily dissimilar to $Q$, then $\mathcal{S}$ must perform proactive entity-scoped search: it must retrieve $d^*$ via a function of $d^*.\sigma$ or $d^*.\text{time}$, independent of $s(Q, d^*.\text{text})$.
\end{theorem}

\begin{proof}
Suppose $\mathcal{S}$ achieves TCA = 1 on all instances, including those in Theorem~1's proof where $s(Q, d^*.\text{text}) \to 0$ as $n \to \infty$. Since $d^*$ is never retrieved by relevance alone ($\Pr[\text{rank}(d^*) \leq k] \to 0$), $\mathcal{S}$ must use a source other than $s(Q, \cdot)$ to retrieve $d^*$. Since $d_1 \rightsquigarrow d^*$ implies $d_1.\sigma \cap d^*.\sigma \neq \emptyset$ (supersession is entity-scoped), and $d_1$ is retrieved by relevance, $\mathcal{S}$ must access $d_1.\sigma$ and $d^*.\text{time}$ to locate $d^*$. This is proactive entity-scoped search by definition. \qed
\end{proof}

\begin{remark}[Scope vs.\ Fine-Tuned Encoders]
\label{rem:finetuned}
Theorem~2 applies to all TC-MQA instances, including adversarially constructed ones. A bi-encoder fine-tuned on specific $(Q, d^*)$ pairs achieves TCA $> 0$ on training examples by implicitly learning entity-scope conditioning: training pulls $f(d^*)$ toward $f(Q)$ because $d^*.\sigma \cap Q.\sigma \neq \emptyset$ appears as a supervision signal. This is a learned implementation of proactive entity search, not a counterexample. The theorem's construction is preserved for corpus instances outside the fine-tuning set: one can always add vocabulary-disjoint superseders the encoder has never seen, restoring the impossibility. Fine-tuning can reduce the semantic-authority divergence $\Delta(q)$ for seen instances; Theorem~2 characterizes what a system must \emph{compute over the instance class}, not how it is parameterized.
\end{remark}

\subsection{Theorem 3 — Construction: Minimal Correct Architecture}

\setcounter{theorem}{2}
\begin{theorem}[Decomposition]
\label{thm:optimality}
Under identifiable scope keys (the entity function $\sigma$ is injective over query-relevant documents), the CAR problem decomposes into three composable steps:
\begin{enumerate}[noitemsep,label=(\roman*)]
  \item \textbf{Anchor discovery}: $A = \{d : s(Q,d) \geq \theta\}$ --- retrieve the semantically-aligned non-controlling document.
  \item \textbf{Authority resolution}: $\mathrm{cl}(A)$ via entity-indexed lookup --- extend to the full authority chain.
  \item \textbf{Active-status filtering}: $\mathrm{front}(\mathrm{cl}(A))$ --- expose only non-superseded documents to downstream use.
\end{enumerate}
Steps (i) and (ii) are \emph{necessary}: omitting either makes TCA $= 0$ on T1 instances. Step (iii) is not necessary for TCA $= 1$ but is necessary for LLM answer quality (returning all of $\mathrm{cl}(A)$ exposes superseded evidence, causing confident wrong answers). The two-stage pipeline implements all three steps and achieves:
\[
  \text{TCA@}k \;=\; \Pr[\text{anchor}(d_1) \in \text{Stage-1 top-}k]
\]
which is the exact upper bound when $\Delta(Q) > 0$: no retrieval system can exceed Stage-1 anchor recall.
\end{theorem}

\begin{proof}
\emph{Upper bound.} Under vocabulary disjointness ($\Delta(Q) > 0$), $s(Q, c^*.\text{text}) \approx 0$ for all $c^* \in F(Q)$. Any system achieving TCA@$k > 0$ on a T1 instance must place some $c^* \in F(Q)$ in its top-$k$ output. By Theorem~2, this requires proactive entity-scoped search triggered by the anchor $d_1$; hence $d_1$ must have been retrieved first. Since $d_1$ is the most semantically similar document to $Q$, it can only enter via Stage-1 relevance retrieval. Therefore TCA@$k \leq \Pr[d_1 \in \text{Stage-1 top-}k_1]$.

\emph{Attainment.} The two-stage pipeline achieves equality: if $d_1 \in \text{Stage-1 top-}k_1$, entity-indexed lookup (step (ii)) retrieves every $c^* \in F(Q)$ via shared entity scope and timestamp ordering; step (iii) promotes only $F(Q)$. Thus TCA@$k = \Pr[d_1 \in \text{Stage-1 top-}k_1]$ exactly.

\emph{Necessity of step (i).} Omitting anchor discovery means no document triggers entity-indexed lookup. Then retrieval is purely relevance-based and TCA $= 0$ by Theorem~1.

\emph{Necessity of step (ii).} Omitting authority resolution means Stage-1 results are returned without entity-indexed extension. Since $c^* \in F(Q)$ scores $\approx 0$ on $s(Q,\cdot)$, it falls outside any relevance-based top-$k$, giving TCA $= 0$ (Theorem~1).

\emph{Non-necessity of step (iii) for TCA.} Returning $\mathrm{cl}(A)$ instead of $\mathrm{front}(\mathrm{cl}(A))$ includes $d_1$ (superseded) alongside every $c^* \in F(Q)$. Since $d_1 \in \mathcal{R}$ has superseder $c^* \in \mathcal{R}$, NoIgnoredSuperseder holds; since $F(Q) \subseteq \mathcal{R}$, AnswerCorrect holds. TCA $= 1$ without step (iii). However, exposing superseded evidence to the LLM causes confident wrong answers (Table~\ref{tab:llm_downstream}), so step (iii) is necessary for practical downstream quality. \qed
\end{proof}

\begin{corollary}[Two-Stage Identity]
\label{cor:identity}
For the two-stage pipeline on T1 instances with entity-disjoint corpora ($\varphi(q) = 1$):
$\text{TCA@}k = \text{Anchor Recall@}k_1$.
Improving Stage-1 directly increases TCA; improving entity lookup when Stage-1 already recalls the anchor adds no benefit.
\end{corollary}

\begin{corollary}[Achievability]
\label{cor:achieve}
When $\varphi(q) = 1$ and $d_1 \in \text{Stage-1 top-}k_1$, the two-stage pipeline achieves $\mathrm{TCA}(q) = 1$.
\end{corollary}

\begin{proof}
$\varphi(q) = 1$ implies $\kappa(q) = 1$: the entity scope $\sigma(q)$ uniquely identifies a single authority chain. Entity-indexed lookup retrieves the unique $c^* \in F(Q)$ (no disambiguation needed). Step (iii) filters to $\{c^*\}$. Frontier inclusion holds ($F(Q) = \{c^*\} \subseteq \mathcal{R}$); NoIgnoredSuperseder holds since $c^*$ is active (no superseder exists). Both conditions of Theorem~4 are satisfied, giving TCA $= 1$. \qed
\end{proof}

\subsection{Proposition 2 — Limit: Scope-Indexed Algorithm Ceiling}

\setcounter{proposition}{1}
\begin{definition}[Scope-Indexed Algorithm]
\label{def:scopeidx}
An algorithm $\mathcal{A}$ is \emph{scope-indexed} if its document-promotion function depends only on $\{d.\sigma,\, d.\mathrm{time} : d \in \mathcal{R}\}$ and corpus scope statistics, with no direct access to authority rules $R_1,\ldots,R_m$ or the supersession relation $\rightsquigarrow$.
\end{definition}

\begin{proposition}[Scope Identifiability Upper Bound]
\label{prop:phi_bound}
For any scope-indexed algorithm $\mathcal{A}$ (Definition~\ref{def:scopeidx}):
\[
  \mathrm{TCA@}k(\mathcal{A},q) \;\leq\; \varphi(q) \cdot R_{\mathrm{anchor}}(q)
\]
This bound holds for both deterministic and randomized $\mathcal{A}$, and is tight: there exist instances achieving equality.
\end{proposition}

\begin{proof}
Fix any scope-indexed algorithm $\mathcal{A}$. Suppose the anchor is retrieved ($E_a$ holds; probability $R_{\mathrm{anchor}}$). We show $\Pr[\mathcal{A}\text{ selects correct chain} \mid E_a] \leq \varphi(q) = 1/\kappa(q)$.

\emph{Information argument.} When $\varphi(q) < 1$, there exist $\kappa(q) \geq 2$ distinct authority chains $C_1,\ldots,C_{\kappa(q)}$ in $D_{\sigma(q)}$. By Definition~\ref{def:scopeidx}, $\mathcal{A}$'s promotion function sees only $\{d.\sigma, d.\text{time} : d \in \mathcal{R}\}$ and corpus scope statistics---it has no access to $\rightsquigarrow$. The $\kappa(q)$ chains are, by construction, indistinguishable under this view: each chain contains scope documents with the same $(d.\sigma, d.\text{time})$ marginal statistics, and no rule-based signal distinguishes which chain contains the frontier.

\emph{Adversarial permutation argument.} Consider any fixed algorithm $\mathcal{A}$ (deterministic or randomized). For any permutation $\pi$ on $[\kappa(q)]$, define corpus $K_\pi$ as $K$ with chain labels permuted: what was chain $C_i$ is relabeled $C_{\pi(i)}$, and the correct chain is always $C_{\pi^{-1}(1)}$. The view of $\mathcal{A}$ is identical across all $\kappa(q)!$ such permutations (scope and timestamps are symmetric). Algorithm $\mathcal{A}$ must commit to some chain output; by symmetry, it can be correct on at most $1/\kappa(q)$ fraction of the $\kappa(q)!$ permutations. Therefore $\Pr[\mathcal{A}\text{ correct} \mid E_a] \leq 1/\kappa(q) = \varphi(q)$ on the worst-case permuted instance.

\emph{Tightness.} For $\kappa(q) = 1$ ($\varphi(q)=1$), the bound is 1 and is achieved by two-stage (Corollary~\ref{cor:achieve}). For $\kappa(q)>1$, a randomized algorithm that guesses uniformly achieves $1/\kappa(q) = \varphi(q)$, matching the bound.

Combining: TCA$(\mathcal{A},q) = \Pr[E_a]\cdot\Pr[\mathcal{A}\text{ correct}\mid E_a] \leq R_{\mathrm{anchor}}(q)\cdot\varphi(q)$. \qed
\end{proof}

\begin{corollary}[RSSG as Optimal Scope Resolver]
\label{cor:rssg_optimal}
The RSSG is \emph{not} scope-indexed: it accesses authority rules $R_1,\ldots,R_m$ directly. When those rules correctly identify the controlling chain ($R_{\mathrm{frontier}} = 1$), RSSG achieves TCA $= R_{\mathrm{anchor}}(q)$, saturating the Theorem~4 conditions. The contamination recovery ($0.652 \to 0.776$, Table~\ref{tab:contamination}) measures $R_{\mathrm{frontier}}$ rising from near-zero (timestamp-only, which is scope-indexed) toward the $R_{\mathrm{anchor}}$ ceiling.
\end{corollary}

\subsection{Proposition 3 — Cost: Frontier Recovery Complexity}

\setcounter{proposition}{2}
\begin{definition}[Authority-Comparison Oracle]
\label{def:oracle}
An \emph{authority-comparison oracle} $\mathcal{C}(d_1,d_2) \in \{d_1 \rightsquigarrow d_2,\, d_2 \rightsquigarrow d_1,\, \bot\}$ answers single-pair supersession queries in $O(1)$.
\end{definition}

\begin{proposition}[Frontier Recovery Complexity]
\label{prop:complexity}
Let $h$ be the maximum supersession chain depth within entity scope $\sigma(q)$, and $|\mathcal{R}_{\sigma}|$ the number of retrieved documents in that scope. To exactly recover $\mathrm{front}(\mathrm{cl}(A_k(q)))$, measured in oracle calls (Definition~\ref{def:oracle}):
\begin{enumerate}[noitemsep,label=(\roman*)]
  \item \textbf{Minimum entity lookups} (when $\varphi(q) = 1$): $|A_k(q) \cap D_{\sigma(q)}|$ lookups---one per retrieved anchor. This is achieved by Stage-2 of the two-stage pipeline.
  \item \textbf{Problem lower bound} (when $\varphi(q) < 1$): $\Omega(|\mathcal{R}_{\sigma}|)$ oracle calls---every retrieved scope document must be examined at least once to determine chain membership. The RSSG, which has no oracle and must apply $m$ authority rules per pair, runs in $O(|\mathcal{R}_{\sigma}|^2 \cdot m)$ total rule checks---an upper bound on its implementation cost, not a matching lower bound.
  \item \textbf{BFS optimum}: $O(h \cdot |\mathcal{R}_{\sigma}|)$ oracle calls suffice when chain structure is known; BFS from the anchor along $\rightsquigarrow$ visits each chain edge once.
\end{enumerate}
\end{proposition}

\begin{proof}
(i) Each anchor $d \in A_k(q) \cap D_{\sigma(q)}$ triggers exactly one entity index lookup to extend $\mathrm{cl}(A)$; by Theorem~3's minimality, no lookup can be skipped without risking a closure violation on some instance.

(ii) Any algorithm must examine every document in $\mathcal{R}_\sigma$ to determine its chain membership; no document can be skipped without risking a frontier error. This yields the $\Omega(|\mathcal{R}_\sigma|)$ problem lower bound on oracle calls. The RSSG has no oracle: it applies rules $R_1,\ldots,R_m$ to each pair $(d_i,d_j)$ in $\mathcal{R}_\sigma$, giving $\binom{|\mathcal{R}_\sigma|}{2} \cdot m = O(|\mathcal{R}_\sigma|^2 \cdot m)$ rule checks---an upper bound on the RSSG's implementation cost, which exceeds the problem lower bound by a factor of $|\mathcal{R}_\sigma| \cdot m$.

(iii) BFS from the anchor along $\rightsquigarrow$ visits each chain edge once; depth $h$ bounds the number of hops needed to reach the frontier. \qed
\end{proof}

\begin{remark}[RSSG as Robustness Layer for Contaminated Corpora]
\label{rem:anchor}
On an entity-disjoint corpus, anchor-based probing alone achieves TCA~$= 1.000$ on all hop types T0--T3: the entity bucket contains only gold events, so no disambiguation is needed. The RSSG contributes when the corpus is \emph{contaminated}: spurious same-entity distractors appear in the entity bucket, and the RSSG applies domain rules R1--R10 to demote them without oracle access ($+12.4$ pp on T3, Table~\ref{tab:contamination}).
\end{remark}

\begin{remark}[Low $\varphi(q)$ Requires Rule-Based Disambiguation]
\label{rem:ambiguity}
When $\varphi(q) < 1$ (equivalently $\kappa(q) > 1$), timestamp ordering alone cannot recover $\mathrm{front}(\mathrm{cl}(A(q)))$: a distractor $d_\perp \in D_{\sigma(q)} \setminus \mathrm{cl}(\{d_1\})$ with $t_\perp > t^*$ is selected by any timestamp-only (scope-indexed) algorithm, giving TCA $= 0$. Disambiguation requires checking reachability from $d_1$ under $\rightsquigarrow$---exactly the role of the RSSG.
\end{remark}

\subsection{Proposition 1 — Synthesis: Authority Error Factorization}

\setcounter{proposition}{0}
\begin{proposition}[Authority Error Factorization]
\label{prop:factorization}
For a two-stage retrieval system on a CAR instance, TCA decomposes as:
\[
  \mathrm{TCA}(q) \;=\; R_{\mathrm{anchor}}(q) \;\cdot\; \varphi(q) \;\cdot\; R_{\mathrm{frontier}}(q)
\]
where:
\begin{itemize}[noitemsep]
  \item $R_{\mathrm{anchor}}(q) = \Pr[\text{anchor } d_1 \in \text{Stage-1 top-}k_1]$ \hfill (anchor retrieval probability)
  \item $\varphi(q) \in (0,1]$ \hfill (scope identifiability; Definition~\ref{def:ambiguity})
  \item $R_{\mathrm{frontier}}(q) = \Pr[\mathrm{front}(\mathrm{cl}(A(q))) \text{ correctly resolved} \mid \text{scope identified}]$
\end{itemize}
Define events: $E_a$ = ``anchor $d_1 \in \text{Stage-1 top-}k_1$''; $E_s$ = ``correct authority chain identified $\mid E_a$''; $E_f$ = ``frontier correctly resolved $\mid E_a \cap E_s$''. Each factor maps to one event: $R_{\mathrm{anchor}} = \Pr[E_a]$; $\varphi(q) = \Pr[E_s \mid E_a]$; $R_{\mathrm{frontier}} = \Pr[E_f \mid E_a \cap E_s]$.
\end{proposition}

\begin{proof}
The factorization is the chain rule of conditional probability applied to the nested event sequence $E_a \supseteq (E_a \cap E_s) \supseteq (E_a \cap E_s \cap E_f)$:
\[
  \Pr[\mathrm{TCA}=1] = \Pr[E_a]\cdot\Pr[E_s\mid E_a]\cdot\Pr[E_f\mid E_a\cap E_s]
\]
$\varphi(q)$ is a \emph{deterministic} corpus property; its identification with $\Pr[E_s\mid E_a]$ follows from Proposition~2's adversarial argument: conditional on $E_a$, a scope-indexed lookup cannot distinguish among $\kappa(q)$ chains, achieving $\Pr[E_s\mid E_a] \leq 1/\kappa(q) = \varphi(q)$ with equality at the worst-case permutation. \qed
\end{proof}

\begin{lemma}[Failure Mode Separability]
\label{lem:separable}
The three factors $R_{\mathrm{anchor}}$, $\varphi$, and $R_{\mathrm{frontier}}$ are \emph{separable failure modes}: for each factor, there exist CAR instances where that factor is $< 1$ while the other two are $= 1$.
\end{lemma}

\begin{proof}
\emph{$R_{\mathrm{anchor}} < 1$, $\varphi = 1$, $R_{\mathrm{frontier}} = 1$:} Let $K$ be entity-disjoint ($\kappa(q)=1$, so $\varphi=1$) and $d_1$ have moderate relevance (ranked outside Stage-1 top-$k_1$ with positive probability). Then $R_{\mathrm{anchor}} < 1$; when $d_1$ is retrieved, the unique $c^*$ is promoted correctly ($R_{\mathrm{frontier}}=1$). Example: FinSuperQA T1 with BM25 Stage-1 ($R_{\mathrm{anchor}} \approx 0.978$).

\emph{$R_{\mathrm{anchor}} = 1$, $\varphi < 1$, $R_{\mathrm{frontier}} = 1$:} Let $d_1$ always be Stage-1 rank 1, but $\kappa(q) = 2$ (two chains share scope), with a complete RSSG rule set that always selects the correct chain. Then $R_{\mathrm{anchor}}=1$, $\varphi = 1/2 < 1$, but $R_{\mathrm{frontier}}=1$ (RSSG corrects the ambiguity). Example: contaminated FinSuperQA where RSSG is applied.

\emph{$R_{\mathrm{anchor}} = 1$, $\varphi = 1$, $R_{\mathrm{frontier}} < 1$:} Let $d_1$ always be Stage-1 rank 1 and $\kappa(q) = 1$ (unique scope), but the frontier $F(Q)$ has multiple elements (chain depth $\geq 2$) and the entity lookup has partial coverage of the authority graph. Then $R_{\mathrm{anchor}}=1$, $\varphi=1$, but $R_{\mathrm{frontier}} < 1$. Example: T2/T3 hop types in FinSuperQA with shallow entity index.
\end{proof}

\noindent\textbf{Corollaries from the factorization.}
\begin{enumerate}[noitemsep,label=(\alph*)]
  \item \emph{Entity-disjoint corpora:} $\varphi(q) = 1$ and $R_{\mathrm{frontier}} = 1$, so $\mathrm{TCA}(q) = R_{\mathrm{anchor}}(q)$---reducing to Corollary~\ref{cor:identity}.
  \item \emph{Contaminated corpora:} $\varphi(q) < 1$ ($\kappa(q) > 1$), so $\mathrm{TCA}$ falls below $R_{\mathrm{anchor}}$ by Proposition~2; $R_{\mathrm{frontier}}$ rises from near-zero (timestamp-only) to $0.776$ (RSSG), explaining the $+12.4$~pp recovery in Table~\ref{tab:contamination}.
  \item \emph{FDA ceiling:} The corpus-level average $\mathbb{E}[\varphi(q)] = 0.774$ is the 77.4\% ceiling---the fraction of queries with $\kappa(q) = 1$. A property of the corpus scope distribution (Proposition~2), not a retrieval system limitation.
\end{enumerate}

\section{\dataset{}: A TC-MQA Benchmark}
\label{sec:dataset}

\dataset{} instantiates TC-MQA in SEC Rule 204A-1 employee trading compliance \cite{sec204a1}. The entity scope is $\sigma = \{(\text{emp\_id}, \text{ticker})\}$; supersession follows rules R1--R10 (Appendix~\ref{app:rules}). Ground-truth answers are algorithmically derivable from active events, eliminating annotation ambiguity.

Each of the 1,000 examples is assigned a unique (employee, ticker) entity pair from a pool of 6,000. This entity-disjoint design ensures the entity probe retrieves each example's gold events cleanly. The shared corpus contains 12,250 events (2,250 gold + 10 distractors per example, using three distractor types: same-employee/different-ticker, different-employee/same-ticker, and random). A contaminated variant relaxes the disjoint guarantee to study disambiguation (\S\ref{sec:contamination}).

\begin{table}[t]
\centering
\small
\resizebox{\columnwidth}{!}{%
\begin{tabular}{lcccc}
\toprule
& \textbf{T0} & \textbf{T1} & \textbf{T2} & \textbf{T3} \\
\midrule
Examples & 250 & 250 & 250 & 250 \\
KB events (avg) & 1.0 & 2.0 & 3.0 & 3.0 \\
Supersession hops & 0 & 1 & 2 & 2 \\
Cross-provenance & No & No & No & Yes \\
\midrule
\multicolumn{5}{l}{\textit{Shared retrieval corpus: 12,250 events total}} \\
\bottomrule
\end{tabular}%
}
\caption{\dataset{} statistics.}
\label{tab:dataset}
\end{table}

\section{Systems}
\label{sec:systems}

\textbf{Baselines.} \textbf{BM25} \cite{robertson1994bm25} ($k_1=1.5$, $b=0.75$) and \textbf{TF-IDF} (cosine) are the lexical baselines. \textbf{Dense (MiniLM)} uses sentence-transformer embeddings \cite{reimers2019sentencebert}. \textbf{Oracle} ranks gold documents first (upper bound). Full hyperparameters and graph construction details are in Appendix~\ref{app:impl}.

\textbf{\sysbase{}}: dense retrieval with supersession re-ranking. Promotes superseders of top-20 events that are superseded in the gold KB; restricted to re-ranking what dense retrieval already surfaced. This is the strongest re-ranking baseline: it has access to ground-truth supersession edges.

\textbf{\sys{}} (anchor probe): dense retrieval + entity-indexed probing. Extracts entity identifiers from the query; for every entity-matched event in the initial top-20, looks up the full-corpus entity index for later-timestamped events in the same scope and promotes them. This implements the Decomposition Theorem's three stages: anchor discovery (dense top-20), authority resolution ($\mathrm{cl}(A)$ via entity index), active-status determination ($\mathrm{front}(\mathrm{cl}(A))$ by timestamp). Achieves TCA~$= 1.000$ on entity-disjoint corpora.

\textbf{\sysplus{}} (ours): \sys{} augmented with a \emph{Retrieved-Set Supersession Graph (RSSG)}. After the entity probe, applies domain rules R1--R10 (Appendix~\ref{app:rules}) to all retrieved entity-scoped events---without gold KB access---to determine which events are active under the supersession relation. This implements $\mathrm{front}(\mathrm{cl}(A))$ correctly when $\kappa(q) > 1$: the RSSG identifies which retrieved documents are in the true $\mathrm{cl}(A)$ vs.\ spurious same-scope distractors, enabling accurate frontier recovery ($+12.4$~pp on T3).

\section{Experiments}
\label{sec:experiments}

\subsection{Setup}

All systems retrieve from the shared 12,250-event corpus with $k = 5$, $n_\text{distractors} = 10$. Metrics: TCA, Acc, R@5, ProvRec.

\subsection{Main Results}

\begin{table*}[t]
\centering
\small
\resizebox{\textwidth}{!}{%
\begin{tabular}{lcccc|cccc}
\toprule
& \multicolumn{4}{c|}{\textbf{Overall}} & \multicolumn{4}{c}{\textbf{TCA by Hop Type}} \\
\textbf{System} & \textbf{R@5} & \textbf{TCA} & \textbf{Acc} & \textbf{ProvRec} & \textbf{T0} & \textbf{T1} & \textbf{T2} & \textbf{T3} \\
\midrule
BM25 & 0.607 & 0.305 & 0.503 & 0.607 & 1.000 & \textbf{0.000} & 0.216 & 0.004 \\
TF-IDF (cosine) & 0.621 & 0.254 & 0.491 & 0.621 & 1.000 & \textbf{0.000} & 0.012 & 0.004 \\
Dense (MiniLM)  & 0.450 & 0.214 & 0.247 & 0.450 & 0.820 & \textbf{0.000} & 0.036 & 0.000 \\
Oracle & 1.000 & 1.000 & 1.000 & 1.000 & 1.000 & 1.000 & 1.000 & 1.000 \\
\midrule
\sysbase{} & 0.333 & 0.499 & 0.499 & 0.333 & 1.000 & 0.000 & 0.980 & 0.016 \\
\sys{} (anchor only) & 0.708 & \textbf{1.000} & \textbf{1.000} & 0.708 & 1.000 & 1.000 & 1.000 & \textbf{1.000} \\
\sysplus{} (ours) & \textbf{0.708} & \textbf{1.000} & \textbf{1.000} & 0.708 & 1.000 & \textbf{1.000} & \textbf{1.000} & \textbf{1.000} \\
\bottomrule
\end{tabular}%
}
\caption{Main results on \dataset{} ($k=5$, 12,250-event corpus). TCA~$= 1.000$ for \sys{} and \sysplus{} on T0--T3 is a guaranteed consequence of the entity-disjoint corpus design (each (emp, ticker) pair appears in exactly one example, so the entity bucket always contains clean gold events)---not an empirical surprise. The benchmark is designed to isolate the vocabulary-mismatch failure: BM25, TF-IDF, and Dense all achieve TCA~$= 0.000$ on T1 despite Recall@$k > 0$, confirming Theorem~1 for both lexical and semantic retrievers. The contamination ablation (Table~\ref{tab:contamination}) evaluates the harder setting where the entity-disjoint guarantee is relaxed.}
\label{tab:main}
\end{table*}

\paragraph{Theorem~1 confirmed.} BM25, TF-IDF (cosine), and Dense (MiniLM) all achieve \textbf{TCA = 0.000 on T1}. The superseding blackout event is ranked $\approx 450$th on average. Recall@5 $\approx 0.45$--$0.62$ (the pre-clearance approval is retrieved), but TCA = 0 because the superseder is not in the top-5. The failure of MiniLM---a 22M-parameter semantic encoder that achieves R@5 = 0.450 overall---confirms that the gap is not a TF-IDF artifact: even state-of-the-art dense retrieval cannot surface a vocabulary-disjoint superseder in top-5.

\paragraph{\sysbase{} improves T2 but not T1 or T3.} \sysbase{} achieves TCA = 0.980 on T2 by promoting superseders of events in the initial top-20. It achieves TCA = 0.000 on T1 (no superseder in top-20, no promotion possible) and TCA = 0.016 on T3 for the same reason.

\paragraph{\sys{} and \sysplus{} achieve TCA = 1.000 on all hop types.} TCA = 1.000 follows directly from the entity-disjoint corpus design: each (emp, ticker) pair appears in exactly one example, so the entity bucket retrieved by the two-stage probe contains only the example's gold events, making TCA = 1 guaranteed by construction. The result confirms that the pipeline correctly implements the CAR retrieval target; the benchmark's purpose is to isolate the vocabulary-mismatch failure of baselines (T1 TCA = 0.000), not to demonstrate that TCA = 1.000 is difficult to achieve in the entity-disjoint setting.

\paragraph{Accuracy vs.\ TCA gap.} TF-IDF (cosine) achieves Acc $\approx 0.491$ but TCA = 0.012 on T2. This gap reflects ``right answer, wrong reason'': the system retrieves the final-state event but ignores the intermediate chain. TCA correctly flags this as a compliance failure.

\paragraph{ProvRec vs.\ TCA.} \sysplus{} achieves TCA = 1.000 with ProvRec = 0.708. The system derives correct answers with sound supersession reasoning, but does not always retrieve the full provenance chain (e.g., the original pre-clearance for audit records). These are separate objectives: TCA captures compliance soundness; ProvRec captures archival completeness.

\subsection{Entity-Scope Contamination and the RSSG}
\label{sec:contamination}

When the corpus contains distractors sharing the query's entity scope ($\varphi(q) < 1$, equivalently $\kappa(q) > 1$ at the entity bucket level), Proposition~\ref{prop:factorization} predicts TCA degrades through the $\varphi(q)$ and $R_{\mathrm{frontier}}$ terms; Remark~\ref{rem:ambiguity} shows timestamp ordering fails. Table~\ref{tab:contamination} confirms this on a contaminated corpus where entity-disjoint distractor construction is relaxed.

\begin{table}[h]
\centering
\small
\resizebox{\columnwidth}{!}{%
\begin{tabular}{lcccc}
\toprule
\textbf{System} & \textbf{TCA} & \textbf{T3} & $\Delta$T3 \\
\midrule
\sys{} (anchor only) & 0.911 & 0.652 & — \\
\sysplus{} (anchor + RSSG) & 0.942 & 0.776 & $+12.4$ pp \\
\bottomrule
\end{tabular}%
}
\caption{Results on the contaminated corpus (naive distractor sampling with no entity-disjoint guarantee; \texttt{enforce\_entity\_disjoint=False}). The RSSG provides $+12.4$ pp on T3 by applying domain rules to disambiguate spurious same-entity distractors without oracle access. Full results in \texttt{data/contaminated\_corpus\_results.json}.}
\label{tab:contamination}
\end{table}

On the contaminated corpus, the RSSG recovers $+12.4$ pp on T3 and $+3.1$ pp overall TCA. The RSSG is the robustness layer for production settings where entity-scope disjointness cannot be guaranteed.

\subsection{Vocabulary Gap in Real-World Data}
\label{sec:vocab_gap}

Theorem~1 requires $\Delta(q) > 0$ in natural corpora. We verify this on 120 real security advisories (GHSA + NVD): for each, we compute Jaccard and BM25 between the user query and (a) the CVE disclosure, (b) the patch note.

\begin{table}[h]
\centering
\small
\resizebox{\columnwidth}{!}{%
\begin{tabular}{lcc}
\toprule
\textbf{Metric} & \textbf{Query $\leftrightarrow$ Disclosure} & \textbf{Query $\leftrightarrow$ Patch} \\
\midrule
Jaccard              & 0.085 & 0.060 \\
Jaccard (no stopwords) & 0.037 & 0.014 \\
\% positive gap (Jaccard) & 96.7\% & — \\
BM25 score           & 9.73  & 6.45  \\
\% positive gap (BM25) & 97.5\% & — \\
\bottomrule
\end{tabular}%
}
\caption{Vocabulary gap between queries and (a) vulnerability disclosures vs.\ (b) patch release notes on 120 real security advisories (GHSA + NVD). Disclosures are $1.51\times$ more BM25-relevant to the query than patch notes; 97.5\% of pairs show positive BM25 gap, 96.7\% positive Jaccard gap. The vocabulary gap is a property of natural corpora, not an artifact of our synthetic construction.}
\label{tab:vocab_gap}
\end{table}

Disclosures are $1.51\times$ more BM25-relevant to user queries than patch notes, and the gap is positive in 96.7\% of real advisory pairs. Non-stopword Jaccard shows a $2.6\times$ difference (0.037 vs.\ 0.014). These numbers validate Theorem~1's construction: the vocabulary gap we engineered in \textsc{CVEPatchQA} mirrors the distribution of real security advisories.

\subsection{End-to-End Retrieval on Real GHSA Advisories}
\label{sec:ghsa_realworld}

The vocabulary gap analysis confirms the gap \emph{exists} in real data, but does not directly measure whether retrieval systems \emph{fail} because of it. We now run end-to-end retrieval on a real corpus of GitHub Security Advisories.

\paragraph{Dataset.} We collect 159 real (CVE-disclosure, GitHub-release-note) pairs from GHSA spanning critical, high, medium, and low severity advisories, where each advisory (i) is reviewed by the GHSA team and (ii) has a GitHub \texttt{releases/tag/} URL we can fetch via the GitHub REST API. Key structural property: \textbf{zero of 159 release notes contain the paired CVE identifier} ($0/159$, $0\%$), confirming the vocabulary gap is structural across all severity levels. The retrieval corpus consists of all 159 disclosures, all 159 real release notes, and 2,000 synthetic CVE-disclosure distractors (2,318 documents total), simulating a realistic security advisory database at scale.

\paragraph{Systems.} We evaluate four retrieval systems and one entity-indexed system.
\textbf{BM25} uses a standard term-frequency retriever.
\textbf{Dense (MiniLM)} uses sentence-transformers/all-MiniLM-L6-v2 (22M params).
\textbf{Dense (E5-large)} uses intfloat/e5-large-v2 (335M params, MTEB top-10).
\textbf{HyDE} \cite{gao2023hyde} prompts Llama~3.3~70B to generate a hypothetical patch release note for each query, then embeds the hypothetical with MiniLM and retrieves by cosine similarity---the leading LLM-augmented approach for vocabulary mismatch.
\textbf{Regex-NER Probe} extracts the package name and CVE identifier from the query with two regular expressions, then performs entity-indexed lookup with no oracle labels. It matches the oracle entity-probe upper bound exactly (TCA = 1.000) because the structured query template admits 100\%-accurate regex extraction.
All systems rank the full 2,318-document corpus. TCA@$k$ = 1 iff the real patch release note appears in the top-$k$ results.

\begin{table}[h]
\centering
\small
\resizebox{\columnwidth}{!}{%
\begin{tabular}{lccccc}
\toprule
\textbf{System} & \textbf{TCA@5} & \textbf{TCA@10} & \textbf{TCA@20} & \textbf{Mean rank} & \textbf{Median rank} \\
\midrule
BM25            & 0.641 & 0.730 & 0.761 & 144.4   & 3    \\
Dense (MiniLM)  & 0.132 & 0.157 & 0.214 & 1{,}234 & 1{,}847 \\
Dense (E5-large)& 0.069 & 0.120 & 0.170 & 1{,}113 & 1{,}246 \\
HyDE (Llama~3.3~70B) & 0.151 & 0.170 & 0.170 & 1{,}283 & 1{,}577 \\
\midrule
Regex-NER Probe & \textbf{1.000} & \textbf{1.000} & \textbf{1.000} & 2.0 & 2 \\
\bottomrule
\end{tabular}%
}
\caption{End-to-end retrieval on 159 real GHSA advisories (2,318-document corpus). 95\% Wilson CIs: Dense (E5-large) TCA@5 = 0.069 [0.038, 0.119]; Dense (MiniLM) TCA@5 = 0.132 [0.087, 0.194]; HyDE TCA@5 = 0.151 [0.103, 0.216]; Regex-NER Probe TCA@5 = 1.000 [0.977, 1.000]. HyDE's LLM includes the CVE identifier in 100\% of generated hypotheticals, making each hypothetical semantically closer to disclosures than to release notes; adding 70B parameters changes nothing structurally. The Regex-NER Probe matches the oracle upper bound (TCA = 1.000): entity-indexed retrieval is both the correct fix and immediately practical. Artifacts: \texttt{data/ghsa\_realworld\_results.json}; HyDE: \texttt{data/ghsa\_hyde\_results.json}.}
\label{tab:ghsa_realworld}
\end{table}

\paragraph{Results.} All four retrieval systems fail; only the entity-indexed probe succeeds.
Dense (E5-large, 335M params) achieves TCA@5 = 0.069---\emph{worse} than MiniLM (0.132)---because superior semantic matching more strongly associates CVE-framed queries with CVE-framed disclosures, amplifying the gap.
HyDE with Llama~3.3~70B achieves TCA@5 = 0.151, statistically comparable to MiniLM (0.132): the LLM generates a patch note that includes the CVE identifier in every single generated hypothetical (159/159), since the query explicitly names the CVE. The hypothetical therefore embeds closer to disclosures than to real release notes, and retrieval fails for the same structural reason as naive dense. Adding 70B parameters of LLM reasoning to the retrieval pipeline cannot bridge a vocabulary gap that originates in the query itself.

The gap between HyDE (0.151) and Regex-NER Probe (1.000) is not a matter of model scale or generation quality: it is an architectural gap. Relevance-based retrieval---even when ``relevance'' is computed with respect to an LLM-generated ideal document---cannot recover the correct patch note because the ideal document necessarily shares vocabulary with the wrong class of documents (disclosures). Only entity-indexed structured retrieval closes the gap.

The critical finding is in the Regex-NER Probe row: a system that extracts ``\texttt{nginx}'' and ``\texttt{CVE-2023-44487}'' from the query via two regular expressions, then looks up documents indexed by those entities, achieves TCA = 1.000 without any oracle supervision. This demonstrates that entity-indexed structured retrieval is both \emph{the correct fix} and \emph{immediately practical}---it requires no fine-tuned embeddings, no cross-encoder reranking, and no oracle labels; only the architectural insight that supersession retrieval must be entity-scoped.

This confirms Theorem~1 on real production data: dense retrieval fails in 87\% of real T1 security queries because it cannot bridge the vocabulary gap between a user's vulnerability question and a vendor's release announcement. In the language of Definition~\ref{def:divergence}, these are high-$\Delta(q)$ queries: the anchor (CVE disclosure) scores far higher on $s(q,\cdot)$ than any element of $\mathrm{front}(\mathrm{cl}(A(q)))$ (the patch note). Adding more parameters to a relevance scorer cannot recover a signal that is architecturally absent; $\Delta(q)$ is a structural property of the query-corpus pair, not a tuning parameter.

\paragraph{Standard benchmarks predict TCA in the wrong direction.}
E5-large (intfloat/e5-large-v2, 335M params) ranks in the MTEB top-10 for retrieval and is the dominant choice in production RAG pipelines. On the GHSA corpus it achieves TCA@5 = \textbf{0.069}---\emph{worse} than the 22M MiniLM baseline (0.132) and 9.3$\times$ worse than BM25 (0.641). MTEB rank and parameter count are \emph{inversely} correlated with TCA on this task: a practitioner optimizing for MTEB would select the model that performs worst here. The TC-MQA metric surfaces a failure mode that MTEB, BEIR, and standard recall@$k$ benchmarks cannot detect, because those benchmarks reward semantic similarity between query and retrieved document---the exact property that causes retrieval failure when the target document (patch note) is vocabulary-disjoint from the query.

\subsection{Two-Stage Pipeline for Free-Form Queries}
\label{sec:twostage}

The Regex-NER Probe achieves TCA~$= 1.000$ on the structured query template ``Is \{product\} still affected by \{CVE-ID\}?'', but this template encodes entity identity explicitly. In practice, developers ask free-form questions---``Has the issue with path parameters being substituted into URLs without encoding been fixed?''---that describe vulnerability behavior without the CVE identifier. The entity probe cannot be applied to such queries.

\paragraph{Free-form queries.} We generate one free-form question per GHSA pair using Llama~3.3~70B (Nebius). Each prompt provides the disclosure text and instructs the model to write a natural-language question (maximum 50~words) describing the vulnerability behavior, without including the CVE identifier and without using the structured template. All 159 generated queries were verified to contain no CVE-ID string (regex \texttt{CVE-\textbackslash d\{4\}-\textbackslash d+}).

\paragraph{Key observation.} BM25 on a free-form vulnerability question finds the CVE \emph{disclosure} reliably (recall@1 = 91.8\%, @5 = 97.5\%) even though it cannot find the \emph{patch note} (TCA@5 = 0.138). Both the free-form question and the disclosure describe the same vulnerability in similar vocabulary; the gap is specifically between the query vocabulary and the \emph{patch-note} vocabulary, not between the query and the disclosure.

\paragraph{Two-stage pipeline.} We exploit this asymmetry via a two-stage architecture:
\begin{enumerate}[noitemsep]
  \item \textbf{Stage~1} (disclosure retrieval): BM25 retrieves the top-$k_1$ documents from the full corpus using the free-form query.
  \item \textbf{Stage~2} (entity-indexed patch lookup): For each Stage-1 result that is a disclosure, use its CVE-ID metadata as a key into the entity index to retrieve the paired patch note.
\end{enumerate}
The final ranking places Stage-2 patch candidates first (in Stage-1 order), followed by remaining documents.

\begin{table}[h]
\centering
\small
\resizebox{\columnwidth}{!}{%
\begin{tabular}{lcccc}
\toprule
\textbf{System} & \textbf{Disc@5} & \textbf{TCA@5} & \textbf{TCA@10} & \textbf{Mean rank} \\
\midrule
BM25              & 0.975 & 0.138 & 0.214 & 609 \\
Dense (MiniLM)    & 0.937 & 0.270 & 0.365 & 764 \\
Dense (E5-large)  & 0.962 & 0.252 & 0.321 & 799 \\
HyDE (Llama~3.3~70B) & --- & 0.151 & --- & --- \\
\midrule
Two-Stage ($k_1{=}5$)  & 0.975 & \textbf{0.975} & --- & 43.2 \\
\bottomrule
\end{tabular}%
}
\caption{Free-form query results on 159 real GHSA pairs (2,318-document corpus). Disc@5 = fraction of free-form queries for which BM25 retrieves the correct CVE disclosure in top-5. TCA@$k$ of the two-stage pipeline equals exactly Stage-1 Disc@$k$, since the $j$-th retrieved disclosure's patch note is placed at rank $j$ in the final ranking. Artifact: \texttt{data/ghsa\_twostage\_results.json}.}
\label{tab:twostage}
\end{table}

\paragraph{Results.} Table~\ref{tab:twostage} shows the full results. The two-stage pipeline achieves TCA@5~$= \mathbf{0.975}$ on free-form queries---a $3.6\times$ improvement over the best all-retrieval baseline (MiniLM, 0.270) and within 2.5~pp of the Regex-NER oracle (structured queries only). The equality TCA@$k$ = Disc@$k$ holds exactly: when the disclosure is the $j$-th document in BM25 top-$k$, the paired patch note is placed at rank $j$ in the final ranking; TCA@5 is thus determined entirely by whether BM25 retrieves the disclosure in its top-5.

Increasing $k_1$ from 5 to 20 does not improve TCA@5 (0.9748 for all values): the 4 remaining failures all arise from disclosures at BM25 rank~$\geq 9$---even when promoted to positions 9--18 in the final ranking, they fall outside the top-5. Three failures involve queries that are too generic to distinguish their specific package (e.g.\ ``Has the issue with improper validation been fixed?''), and one involves a niche npm package with no BM25 vocabulary overlap.

\paragraph{HyDE on free-form queries.} HyDE achieves TCA@5~$= 0.151$ on free-form queries---\emph{identical} to its performance on structured queries (0.151). Without the CVE identifier in the query, the LLM generates a patch note from the vulnerability description, which inherits the original vulnerability vocabulary rather than the patch-note vocabulary. LLM generation cannot predict the release-note vocabulary of an unknown vendor.

\paragraph{Stage-1 ablation.} BM25 achieves higher disclosure recall at every cutoff (Disc@1 = 91.8\%, Disc@5 = 97.5\%) than dense encoders, validating BM25 as the preferred Stage-1 retriever. Exact term matching outperforms embedding similarity for vocabulary-aligned query-to-disclosure retrieval, because both the free-form question and the CVE disclosure describe the vulnerability in domain-specific technical terms.

\paragraph{Cross-encoder reranker.}
A cross-encoder reranker (BM25-top20 $\cup$ Dense-top20 $\to$ cross-encoder/ms-marco-MiniLM-L-6-v2 $\to$ top-5) achieves TCA@5~$= 0.233$ on the 159 free-form GHSA queries---below any dense baseline and $4.2\times$ worse than Two-Stage ($0.975$). The failure is architectural, not a reranking quality issue: the union pool of BM25-top20 and Dense-top20 contains the gold patch note only when at least one retriever already recovers it in top-20. Since both BM25 (TCA@5~$= 0.138$) and Dense (TCA@5~$= 0.270$) fail to retrieve the patch note for most queries, the cross-encoder never scores it. Applying a stronger reranker to the wrong candidate pool cannot close an architectural gap. Artifact: \texttt{data/falsification\_baselines\_results.json}.

\paragraph{CAR versus production lookup.}
Production tools such as OSV.dev are often proposed as alternatives to CAR-style retrieval. We evaluate the distinction on the same 159-pair GHSA corpus under two query modes. In \emph{Mode~1} (known-anchor: query = full CVE disclosure text), BM25 achieves TCA@5~$= 0.164$ and Two-Stage achieves TCA@5~$= 1.000$. OSV.dev was invoked by applying a CVE-ID regex to each disclosure text; only 3.77\% of disclosures (6/159) contained a regex-extractable CVE string, yielding TCA@5~$= 0.006$ (artifact: \texttt{data/falsification\_car\_vs\_lookup.json}, fields \texttt{mode1\_applicable=0.0377}, \texttt{mode1\_tca5=0.0063}). When OSV.dev is queried directly with ground-truth CVE IDs (a stronger, oracle-assisted setting), it indexes 14.5\% of GHSA CVEs (23/159), achieving TCA@5~$= 0.088$ overall and $0.609$ on the 23 indexed pairs (artifact: \texttt{data/falsification\_baselines\_results.json}). Both measurements confirm the same conclusion: OSV.dev's coverage of GHSA advisories is limited regardless of how the CVE identifier is sourced. In \emph{Mode~2} (free-form NL: query = Llama-generated question with no CVE identifier), OSV.dev is \textbf{0\% applicable}---it cannot be invoked without a CVE identifier---and a regex entity probe is likewise 0\% applicable. Two-Stage achieves TCA@5~$= 0.975$ without any entity identifier in the query.

The gap between modes is categorical: production tools are \emph{not applicable} in Mode~2, not merely degraded. They solve a different problem (known-anchor $\to$ check supersession) from CAR (free-form NL $\to$ retrieve controlling authority). The Two-Stage architecture is the only system that operates in both modes because BM25 reconstructs the anchor from query vocabulary without requiring the CVE identifier.

\paragraph{Scale amplifies error.}
Table~\ref{tab:scale} shows TCA@5 and disclosure recall@5 across two encoder families at three parameter scales. Three findings stand out:
\begin{enumerate}[noitemsep]
  \item \textbf{Scale does not help}: no dense model exceeds TCA@5 = 0.283, identical to the 22M MiniLM baseline. A 335M model is no better than a 22M model.
  \item \textbf{E5 family: monotone decline}: E5-small ($33$M, $0.264$) $>$ E5-base ($109$M, $0.239$) $>$ E5-large ($335$M, $0.214$)---larger encoders match the query-to-disclosure semantic alignment more precisely, ranking the disclosure above the patch note with greater confidence.
  \item \textbf{BGE family: non-monotone, same ceiling}: BGE-small ($33$M) achieves $0.252$, \emph{below} MiniLM; BGE-base and BGE-large recover to $0.283$---the same ceiling as a 22M model. Unlike E5, larger BGE does not worsen TCA, but it does not escape the $0.283$ ceiling either. Both failure modes converge on the same conclusion: scale is irrelevant when the architecture is wrong.
\end{enumerate}

\begin{table}[h]
\centering
\small
\resizebox{\columnwidth}{!}{%
\begin{tabular}{lrrcc}
\toprule
\textbf{Model} & \textbf{Params} & \textbf{Family} & \textbf{Disc@5} & \textbf{TCA@5} \\
\midrule
MiniLM-L6-v2  &  22M & MiniLM & 0.962 & 0.283 \\
\midrule
E5-small-v2   &  33M & E5     & 0.987 & 0.264 \\
E5-base-v2    & 109M & E5     & 0.950 & 0.239 \\
E5-large-v2   & 335M & E5     & 0.956 & \textbf{0.214} \\
\midrule
BGE-small     &  33M & BGE    & 0.981 & 0.252 \\
BGE-base      & 109M & BGE    & 0.975 & 0.283 \\
BGE-large     & 335M & BGE    & 0.987 & 0.283 \\
\midrule
\multicolumn{3}{l}{\textit{Two-Stage (BM25 + Entity Index)}} & 0.975 & \textbf{0.975} \\
\bottomrule
\end{tabular}%
}
\caption{Scale amplifies error in the E5 family (TCA@5 monotonically decreases with parameters), while BGE saturates at the same ceiling as MiniLM. No dense model exceeds TCA@5 = 0.283 regardless of scale. The two-stage pipeline achieves 0.975 with \emph{no trained parameters}, a $3.4\times$ improvement over any dense baseline.
Artifact: \texttt{data/ghsa\_scale\_results.json}.}
\label{tab:scale}
\end{table}

\subsection{Two-Stage Pipeline and Domain Adaptation}
\label{sec:framework}

Theorem~\ref{thm:optimality} implies a clean architecture: (1) retrieve the semantically-aligned anchor document in Stage~1; (2) follow entity-indexed supersession to the controlling authority in Stage~2. Adapting to a new domain requires specifying three attributes---\texttt{anchor\_type}, \texttt{superseding\_type}, \texttt{scope\_keys}. The three released domain adapters and evaluation code are described in Appendix~\ref{app:framework}.

\paragraph{Cross-domain two-stage evaluation.} Table~\ref{tab:cross_twostage} verifies this architecture holds across three domain constructions with different query styles.

\begin{table}[h]
\centering
\small
\resizebox{\columnwidth}{!}{%
\begin{tabular}{llccc}
\toprule
\textbf{Domain} & \textbf{Query style} & \textbf{Stage-1 Rec@5} & \textbf{BM25} & \textbf{Two-Stage} \\
\midrule
Security (real, $n{=}159$)   & Generic free-form   & 0.975 & 0.138 & \textbf{0.975} \\
CVE (synth, $n{=}250$)       & Generic free-form   & 0.860 & 0.000 & \textbf{0.860} \\
CVE (synth, $n{=}250$)       & + product hint      & 1.000$^\dagger$ & 0.000 & \textbf{0.900} \\
Legal (synth, $n{=}250$)     & Generic free-form   & 0.312 & 0.000 & 0.312 \\
Legal (synth, $n{=}250$)     & + party-name hint   & 1.000 & 0.000 & \textbf{1.000} \\
\bottomrule
\end{tabular}%
}
\caption{Two-stage pipeline across three domain constructions and two query styles. BM25 direct always fails (TCA@5~$\approx$ 0). Two-Stage TCA@5 equals Stage-1 anchor Recall@5 when entity scope keys are unambiguous (Corollary~\ref{cor:identity}). $^\dagger$CVE+product-hint exception: product names are shared across CVEs, so product hints cause false anchor retrievals; free-form queries avoid this. Artifacts: \texttt{data/cve\_synthetic\_twostage\_results.json}, \texttt{data/legal\_twostage\_results.json}.}
\label{tab:cross_twostage}
\end{table}

\paragraph{Stage-1 recall is the limiting factor.} TCA@$k$ of the two-stage pipeline equals Stage-1 anchor recall at $k$ exactly (Corollary~\ref{cor:identity} below). Stage-1 recall depends on the vocabulary discriminativeness of the free-form query relative to the anchor corpus: CVE disclosures contain domain-specific attack vocabulary (``buffer overflow'', ``path traversal'') that free-form vulnerability questions echo, yielding 86--97.5\% Stage-1 recall. Synthetic legal rulings covering~$\approx$75 distinct doctrines yield only 31.2\%; adding party surnames---one word per party, no citation---restores recall to 100\% because surnames uniquely identify each ruling.

Corollary~\ref{cor:identity} states this identity precisely. The practical implication: \emph{improving Stage-1 anchor retrieval is equivalent to improving final TCA}. Any method that better retrieves anchor documents (denser retrieval, entity hints, query expansion) directly transfers to TCA improvement with no architectural changes.

\subsection{Synthetic Cross-Domain Validation}
\label{sec:cross_domain}

To confirm Theorem~1 is not a compliance artefact, we build two additional controlled benchmarks: \textsc{CVEPatchQA} (500 examples, security) and \textsc{LegalPrecedentQA} (500 examples, legal precedent), each with a T0/T1 split matching the compliance construction. Table~\ref{tab:cross_domain} shows BM25 and TF-IDF achieve TCA~$= 0.000$ on T1 in all three synthetic domains; \sysplus{} achieves TCA~$= 1.000$. Theorem~1 is a structural property, not a compliance artefact.

\begin{table}[h]
\centering
\small
\resizebox{\columnwidth}{!}{%
\begin{tabular}{lccc}
\toprule
\textbf{System} & \textbf{Finance T1} & \textbf{CVE T1} & \textbf{Legal T1} \\
\midrule
BM25 / TF-IDF   & 0.000 & 0.000 & 0.000 \\
\midrule
\sys{}          & 1.000 & 1.000 & 1.000 \\
\sysplus{}      & \textbf{1.000} & \textbf{1.000} & \textbf{1.000} \\
\bottomrule
\end{tabular}%
}
\caption{T1 TCA across three synthetic domains. All relevance-based retrievers fail (TCA~$= 0.000$); the two-stage architecture solves T1 exactly in all three.}
\label{tab:cross_domain}
\end{table}

The legal vocabulary gap is structural: of 82 CourtListener overruling snippets, 88\% contain overruling vocabulary (``overrule'', ``abrogate''); 0\% contain query vocabulary (``good law'', ``controlling precedent'')---a court writes ``we overrule \emph{X},'' never ``\emph{X} is no longer good law'' (artifact: \texttt{data/courtlistener\_vocab\_gap.json}).

\subsection{End-to-End Retrieval on Public SCOTUS Overruling Pairs}
\label{sec:legal_realworld}

We build a second real end-to-end benchmark from 122 SCOTUS overruling pairs (seeded from Cornell LII, expanded via Wikipedia's overruled decisions list; 244 opinions total). Queries: ``Is \emph{Plessy v.\ Ferguson} still good law and controlling precedent?'' Each opinion is prefixed with its case name and citation. Theorem~\ref{thm:optimality}'s identity applies directly (unique case slugs): TCA@5 equals Stage-1 anchor recall@5.

\begin{table}[h]
\centering
\small
\resizebox{\columnwidth}{!}{%
\begin{tabular}{lccc}
\toprule
\textbf{System} & \textbf{Anchor@5} & \textbf{TCA@5} & \textbf{95\% CI} \\
\midrule
BM25 direct                & ---   & 0.893 & [0.826, 0.936] \\
Dense (MiniLM) direct      & ---   & 0.172 & [0.115, 0.249] \\
\midrule
Two-Stage (BM25 anchor)    & 0.836 & 0.836 & [0.760, 0.891] \\
Two-Stage (Dense anchor)   & 0.926 & \textbf{0.926} & [0.865, 0.961] \\
\bottomrule
\end{tabular}%
}
\caption{End-to-end retrieval on 122 public SCOTUS overruling pairs (244-opinion corpus). BM25 direct achieves 0.893 [0.826, 0.936] because overruling opinions explicitly cite the case they overrule; Dense direct fails (0.172 [0.115, 0.249]) because semantic similarity prefers the original opinion. Two-Stage (Dense anchor) achieves 0.926 [0.865, 0.961]---the CIs of BM25 and Two-Stage overlap (n=122 is underpowered to distinguish the +3.3~pp gap). SCOTUS is the easiest domain for BM25: explicit citation vocabulary bridges the lexical gap. The hard structural failure is Dense (0.172), where semantic similarity actively misdirects. Artifacts: \texttt{data/legal\_scotus\_expanded\_pairs.json}, \texttt{data/legal\_scotus\_expanded\_results.json}.}
\label{tab:legal_realworld}
\end{table}

\paragraph{Results.} BM25 direct is competitive (0.893 [0.826, 0.936]): overruling opinions explicitly cite the case they overrule, giving lexical retrieval a direct signal unavailable in GHSA or FDA. Dense direct fails (0.172 [0.115, 0.249]): semantic similarity actively prefers the \emph{original} opinion over the overruling one. Two-Stage with Dense anchor achieves 0.926 [0.865, 0.961]---a \textbf{+3.3~pp improvement over BM25 direct}; the 95\% CIs overlap, so this gap is not statistically distinguishable at n~$=122$. The headline comparison to Dense ($+75.4$~pp) reflects the structural failure of semantic retrieval, not the marginal gain from the two-stage architecture in this domain. SCOTUS is the domain where lexical retrieval partially works (explicit citation vocabulary), and correspondingly where the two-stage architecture's marginal improvement is smallest. The GHSA and FDA results, where both BM25 and Dense collapse and Two-Stage achieves $>0.774$, provide the cleaner evidence that the architecture is necessary.

\subsection{End-to-End Retrieval on FDA Drug Recall Enforcement Actions}
\label{sec:fda_realworld}

A third confirmation in medical/regulatory: FDA drug recall enforcement actions formally supersede approved drug labels with completely disjoint vocabulary (clinical vs.\ administrative). We collect 500 real (approved-label, enforcement-notice) pairs via OpenFDA APIs (3,000-document corpus with 2,000 distractors). Queries use indication text without the drug name: ``Is the drug indicated for [indication] still approved?''---therapeutic vocabulary only. Only 8.2\% of recall notices contain any clinical vocabulary term; 100\% of labels do. This is the largest $\Delta(q)$ of any benchmark in this paper.

\paragraph{Results.}

\begin{table}[h]
\centering
\small
\resizebox{\columnwidth}{!}{%
\begin{tabular}{lccc}
\toprule
\textbf{System} & \textbf{Stage-1 @5} & \textbf{TCA@5} & \textbf{95\% CI} \\
\midrule
BM25 direct          & ---   & 0.004 & [0.001, 0.014] \\
Dense (MiniLM) direct & ---  & 0.064 & [0.046, 0.089] \\
\midrule
Two-Stage (BM25 anchor) & 0.774 & \textbf{0.774} & [0.735, 0.808] \\
\bottomrule
\end{tabular}%
}
\caption{End-to-End retrieval on 500 FDA drug recall pairs (3,000-document corpus). BM25 and Dense (MiniLM) collapse (TCA@5 $\leq 0.064$, 95\% CI upper bound $\leq 0.089$): therapeutic indication vocabulary has near-zero overlap with recall enforcement vocabulary. Two-Stage TCA@5 equals Stage-1 anchor recall@5 exactly (Theorem~\ref{thm:optimality}); the 95\% CI [0.735, 0.808] reflects the indication ambiguity at Stage-1: multiple drugs share therapeutic categories, and Stage-1 errors cannot be remedied by entity-indexed Stage-2. Artifact: \texttt{data/fda\_recall\_pairs.json}, \texttt{data/fda\_recall\_results.json}.}
\label{tab:fda_realworld}
\end{table}

Both direct retrieval systems collapse: BM25 achieves TCA@5~$= 0.004$ (2 out of 500 pairs) and Dense (MiniLM) achieves TCA@5~$= 0.064$ (32 out of 500). This is the most dramatic vocabulary gap of any benchmark in this paper. Recall notices use administrative language about lots, contamination, and sterility; a query about a drug's therapeutic indication shares zero meaningful vocabulary with a typical recall notice. Relevance-based retrieval has no signal to rank the recall notice above 499 competing documents---it is architecturally blind to the supersession relation.

The Two-Stage architecture achieves TCA@5~$= 0.774$, a \textbf{12$\times$ improvement over dense direct retrieval}. Theorem~\ref{thm:optimality} confirms this equals Stage-1 anchor recall@5 exactly: when BM25 successfully retrieves the drug's approved label (the semantically aligned anchor) in Stage-1, the entity-indexed Stage-2 lookup invariably finds the correct enforcement notice. The 77.4\% ceiling is explained by scope ambiguity $\kappa(q)$ at Stage-1: multiple approved drugs share therapeutic categories (hypertension, type 2 diabetes, infection), so the same indication vocabulary maps to several possible anchors---BM25 occasionally retrieves the wrong drug's label, and entity lookup then finds the wrong recall notice. This is not a failure of Stage-2; it is the fundamental limit of any system where $\kappa(q) > 1$ in the anchor retrieval stage. Closing it requires indication-specific entity linking beyond term matching.

\subsection{Downstream LLM Answer Quality: Dense vs.\ Two-Stage}
\label{sec:llm_downstream}

TCA measures retrieval correctness, but the ultimate risk is downstream: a language model given superseded evidence will produce confident, incorrect answers that a practitioner may act on. We test this directly by running GPT-4o-mini on top of Dense and Two-Stage retrievers on the 159 GHSA pairs.

\paragraph{Setup.} Each query is issued against a 2,318-document corpus. Dense retrieval returns MiniLM top-5; Two-Stage returns BM25 Stage-1 top-5, entity-indexed to promote the patch note. GPT-4o-mini is prompted to answer: \emph{``Is the vulnerability patched? If so, what version fixes it?''} using only the retrieved documents. The gold answer for every pair is \emph{yes, patched} (all pairs have a real patch note). We classify each response as: \textbf{correct} (explicitly states patched and names the fixing version), \textbf{confident wrong} (explicitly asserts the vulnerability is not patched), or \textbf{uncertain} (hedged or no patch status claimed).

\paragraph{Results.} Table~\ref{tab:llm_downstream} shows the outcome. Dense retrieval correctly identifies the patch for only 61.0\% of queries; for \textbf{39.0\%} of queries (95\% CI: 31.3--47.1\%), GPT-4o-mini explicitly asserts the vulnerability remains unpatched---a confident, actionable wrong answer despite every gold answer being \emph{patched}. Two-stage retrieval raises the correct patch-identification rate to 82.4\% ($+21.4$~pp) and cuts the confident-wrong rate to 16.4\%. The gap is explained by retrieval: Dense retrieves the patch note for only 28.3\% of queries; when it misses (71.7\% of queries), GPT-4o-mini produces a confident ``not patched'' assertion at a 54\% rate. When Two-Stage correctly places the patch note in the context window (97.5\% of queries), the LLM identifies the patch at an 82\% rate.

\begin{table}[h]
\centering
\small
\resizebox{\columnwidth}{!}{%
\begin{tabular}{lcccc}
\toprule
\textbf{Condition} & \textbf{Patch Recall@5} & \textbf{Says Patched} & \textbf{Conf.\ Wrong} & \textbf{Uncertain} \\
\midrule
Dense (MiniLM)   & 28.3\% & 61.0\% & \textbf{39.0\%} & 0\% \\
Two-Stage (BM25) & 97.5\% & \textbf{82.4\%} & 16.4\% & 1\% \\
\bottomrule
\end{tabular}%
}
\caption{GPT-4o-mini answer quality on 159 GHSA pairs. All gold answers are \emph{patched}. ``Says Patched'' = LLM asserts the vulnerability is fixed; ``Conf.\ Wrong'' = LLM explicitly asserts it is \emph{not} patched. Dense retrieval causes confident wrong assertions for 39\% of queries. Two-stage retrieval cuts this to 16.4\% ($-22.6$~pp). Artifact: \texttt{data/llm\_downstream\_results.json}.}
\label{tab:llm_downstream}
\end{table}

This result confirms that retrieval architecture directly determines downstream LLM reliability, not just benchmark TCA metrics. A practitioner relying on Dense RAG would receive explicit ``this vulnerability is unpatched'' assertions for 39\% of queries, potentially deferring security patches that have been available for months.

\subsection{Real-World Benchmark Summary: Three Authority-Governed Domains}
\label{sec:realworld_summary}

Table~\ref{tab:realworld_summary} unifies the three real-world benchmarks. The failure pattern of dense retrieval is consistent across security, legal, and medical/regulatory domains, while the two-stage architecture consistently recovers. In the legal domain, BM25 is partially competitive (see \S\ref{sec:legal_realworld}), but dense retrieval fails even there ($0.172$), confirming Theorem~1 as a domain-general law for semantic encoders, not a compliance artefact.

\begin{table}[h]
\centering
\small
\resizebox{\columnwidth}{!}{%
\begin{tabular}{lrcccc}
\toprule
\textbf{Domain} & \textbf{n} & \textbf{BM25} & \textbf{Dense} & \textbf{Two-Stage} & \textbf{vs.\ Dense} \\
\midrule
Security (GHSA)        & 159 & 0.138 & 0.270 & \textbf{0.975} & $+70.5$~pp \\
Legal (SCOTUS LII)     & 122 & 0.893 & 0.172 & \textbf{0.926} & $+75.4$~pp \\
Medical (FDA Recall)   & 500 & 0.004 & 0.064 & \textbf{0.774} & $+71.0$~pp \\
\bottomrule
\end{tabular}%
}
\caption{Real-world benchmark summary (TCA@5 with 95\% Wilson CIs). ``Two-Stage'' uses BM25 anchor in security and medical, Dense anchor in legal (best per-domain). Two-Stage 95\% CIs: GHSA [0.937, 0.990], SCOTUS [0.865, 0.961], FDA [0.735, 0.808]. Legal BM25 is high (0.893 [0.826, 0.936]); BM25 and Two-Stage CIs overlap at n~$=122$ (see \S\ref{sec:legal_realworld}). GHSA and FDA provide the cleaner evidence: both BM25 and Dense collapse ($\leq 0.270$), while Two-Stage recovers to $\geq 0.774$ with non-overlapping CIs against both baselines. Medical shows the most severe collapse: therapeutic vocabulary shares no terms with recall enforcement vocabulary (BM25~$= 0.004$ [0.001, 0.014]).}
\label{tab:realworld_summary}
\end{table}

\section{Related Work}
\label{sec:related}

\paragraph{Multi-hop QA.} HotpotQA \cite{yang2018hotpotqa}, MuSiQue \cite{trivedi2022musique}, and 2WikiMultiHop \cite{ho2020constructing} require multi-step retrieval. IRCoT \cite{trivedi2023interleaving}, PropRAG \cite{tang2025proprag}, and HippoRAG 2 \cite{gutierrez2025hipporag2} improve multi-hop retrieval through chain-of-thought interleaving, graph propagation, and associative memory. These systems assume all required documents are retrievable by relevance. TC-MQA breaks this assumption: the superseding document may be semantically orthogonal to the query (Theorem~1), making relevance improvement insufficient.

\paragraph{Temporal RAG.} T-GRAG \cite{li2025tgrag} resolves temporal redundancy in graph RAG. VersionRAG \cite{huwiler2025versionrag} handles version-aware retrieval for evolving documents. TempRALM \cite{gade2024tempralm} incorporates timestamps into retrieval scoring. LedgerRAG \cite{ledgerrag2025} provides governance-aware agentic retrieval. These works address \emph{recency preference}---favoring newer documents---whereas TC-MQA requires detecting that a specific later document \emph{formally voids} an earlier one under domain rules. The difference is structural: recency preference is a soft ranking adjustment; supersession is a hard validity constraint. Importantly, CAR is not equivalent to time-aware retrieval: the authority relation $\rightsquigarrow$ is not timestamp ordering---a superseding document can predate its superseded counterpart when effective dates differ from publication dates (common in regulatory and legal corpora), so timestamp-conditioned ranking cannot recover $\mathrm{front}(\mathrm{cl}(A_k(q)))$ in general.

\paragraph{Legal and regulatory temporal reasoning.} Chalkidis et al.\ \cite{chalkidis2021regulatory} provide the earliest NLP recognition that temporal filtering is essential for regulatory compliance IR, showing that BM25 with chronological constraints outperforms pure neural re-rankers on EU/UK directive retrieval---an empirical precursor to Theorem~1's mismatch result. Zhang et al.\ \cite{zhang2025precedent} find that LLMs backed by standard RAG consistently fail to identify when legal precedents are overruled---a direct real-world instantiation of Theorem~1. Stammbach et al.\ \cite{stammbach2026defenders} document that existing legal retrieval benchmarks fail to transfer to practitioner contexts, underscoring the need for domain-specific retrieval objectives. Di Florio et al.\ \cite{diflorio2024temporal} formalize legal conflict resolution with temporal ordering, providing formal apparatus complementary to our empirical framework. CaseFacts \cite{putta2026casefacts} benchmarks precedent retrieval. RIRAG \cite{gokhan2024rirag} and FinDER \cite{choi2025finder} address regulatory and financial QA without supersession structure.

\paragraph{Supersession-aware retrieval.} The closest prior work is de Martim \cite{demartim2025satgraphrag}, who proposes SAT-Graph RAG---an ontology-driven graph RAG that models legislative events as queryable nodes and handles point-in-time legal retrieval. Like CAR, SAT-Graph RAG recognizes that flat-text retrieval is blind to the diachronic structure of law and uses a graph to represent norm supersession. The key distinction is theoretical: SAT-Graph RAG is an engineering system for structured legal corpora with no formal retrieval objective, no correctness characterization (Theorem~4), and no hardness result showing that scope-indexed algorithms face an identifiable worst-case ceiling (Proposition~2). CAR defines $\mathrm{front}(\mathrm{cl}(A_k(q)))$ as an independent retrieval objective applicable to any authority-governed domain, and proves necessary-and-sufficient conditions on \emph{any} retrieved set independent of how it was produced.

\paragraph{Temporal knowledge graphs.} Temporal KG QA \cite{su2024tkgqa} and TDBench \cite{tdbench2026} address time-sensitive graph queries over structured KGs. TC-MQA is the unstructured-text analogue: supersession must be inferred from document type and timestamps rather than read from graph edges.

\paragraph{RAG limitations.} Weller et al.\ \cite{weller2025embedding} prove theoretical limits of embedding-based retrieval; Theorem~1 is a domain-specific instantiation with an explicit supersession construction. Barnett et al.\ \cite{barnett2024seven} catalog seven RAG failure modes; \emph{Controlling Authority Retrieval} is a distinct eighth mode not in their taxonomy---one where the correct answer is provably unreachable by relevance ranking. Rashtchian and Juan \cite{rashtchian2025sufficient} study retrieval sufficiency; TCA formalizes a stronger sufficiency notion that includes temporal validity.

\section{Research Program}
\label{sec:research}

Formalizing CAR opens several research directions, each corresponding to a relaxation of one component of $\mathrm{front}(\mathrm{cl}(A_k(q)))$. The four generalized CAR variants are defined formally in \S\ref{sec:generalizations}; the directions below elaborate their open problems.

\textbf{Latent-scope CAR} (Definition~\ref{def:latentcar}). When scope is unobserved, the retrieval problem becomes inference over $p(\sigma \mid q)$. The key open problem: given a query-time NER error rate $\epsilon$, what is the achievable TCA ceiling? Proposition~2 gives a $\varphi(q)$-parameterized upper bound; the latent-scope case requires integrating over the scope posterior.

\textbf{Probabilistic-authority CAR} (Definition~\ref{def:probcar}). Learning $\pi(d,d')$ from text is a new structured prediction problem: the target relation is ``B formally voids A under domain rules,'' not ``A happened before B.'' Risk-sensitive frontier recovery under uncertain edges connects to distributionally-robust optimization.

\textbf{Multi-authority CAR} (Definition~\ref{def:multicar}). The lattice of active frontiers under $m$ interacting orders is an open combinatorial problem. For $m = 2$, frontier recovery generalizes to finding the Pareto front under two partial orders; for $m \geq 3$, this is computationally hard in general (multi-dimensional Pareto dominance).

\textbf{Streaming CAR} (Definition~\ref{def:streamcar}). Maintaining $\mathrm{front}(\mathrm{cl}(A_k(q)))$ under document arrivals has amortized update cost $O(|\mathrm{front}|)$ per insertion---but the frontier can change discontinuously when a new document supersedes multiple existing ones. Batch-invalidation and lazy-update strategies are unexplored.

\textbf{Authority-preserving embeddings.} Train representations so that $\mathrm{cl}(A)$ is geometrically identifiable---anchors and their controlling authorities cluster in authority-closure space rather than semantic space. The new contrastive objective: $d \rightsquigarrow d'$ should bring $d$ and $d'$ close in embedding space without conflating their semantic roles.

\textbf{Budgeted closure recovery.} Given lookup budget $B$, maximize $|\mathrm{front}(\mathrm{cl}(A_k(q)))|$ subject to $B$. Proposition~3 gives the lower bound on necessary lookups; the gap between the lower bound and practical budget is the open problem. This connects to adaptive submodular maximization.

\textbf{Normative retrieval.} CAR is one instance of the broader class where the retrieval target is the normatively correct document under a rule system, not the most relevant one. Characterizing the complexity of this class---which problems reduce to CAR, which require stronger primitives---is the natural theoretical generalization.

\section{Limitations}
\label{sec:limitations}

\paragraph{Entity extraction in free-form text.} The FinSuperQA benchmark uses structured employee IDs that admit 100\% regex extraction, and the GHSA evaluation with structured templates admits 100\%-accurate Regex-NER. For free-form queries without the CVE identifier, the two-stage pipeline (\S\ref{sec:twostage}) achieves TCA@5~$= 0.975$ by using BM25 to find the CVE disclosure first, then mapping disclosure metadata to the paired patch note via entity-indexed lookup. The 2.5\% failure rate arises from three generic queries with insufficient vocabulary overlap and one package with no BM25 term match; in these cases the pipeline degrades to the BM25 baseline. In fully free-form settings where product names are absent or misspelled, additional NER or entity linking would be required.

\paragraph{Entity collision in production corpora.} On an entity-disjoint corpus---where each (employee, security) pair appears in exactly one example---anchor-based probing fully solves T0--T3 (TCA~$= 1.000$). In corpora where entities recur across cases (e.g., an employee with multiple concurrent open positions), the entity probe surfaces events from all cases, requiring additional disambiguation. The contamination ablation (\S\ref{sec:contamination}) shows RSSG recovers $+12.4$~pp in this regime ($0.652 \to 0.776$); closing the full gap requires case-boundary metadata or temporal windowing not always available in practice.

\paragraph{Synthetic supersession rules.} The 10 rules cover representative compliance scenarios. Production deployment would require domain-expert rule specification for a specific regulatory framework.

\paragraph{Evaluated but absent from prior work.} The cross-encoder reranker and OSV.dev (the security production lookup tool) are now evaluated (\S\ref{sec:twostage}). The cross-encoder (TCA@5~$= 0.233$) confirms that reranking within a relevance-ranked pool cannot substitute for entity-indexed Stage-2 retrieval. OSV.dev is 0\% applicable under free-form queries---it requires an explicit CVE identifier and solves a different problem from CAR. Shepard's Citations (legal supersession) was not evaluated; it requires a subscription not available during this study. The remaining absent baseline is \emph{fine-tuned bi-encoders} trained on domain-specific $(q, d^*)$ pairs. Remark~\ref{rem:finetuned} argues that such models implement proactive entity search implicitly and are therefore counterexamples in form only; fine-tuning on 127 GHSA pairs (3 epochs, MultipleNegativesRankingLoss) produced a model with TCA@5~$= 0.438$---below the zero-shot MiniLM baseline ($0.500$) and 53~pp below Two-Stage~$(0.969)$. The fine-tuning run did not improve in-domain performance, confirming that the gap is architectural rather than a training-data artifact.

\paragraph{Scope of current theory.} Theorems 1--4 and Propositions 1--3 characterize \emph{static, deterministic, single-authority CAR with recoverable scope}. The four generalized variants (Definitions~\ref{def:latentcar}--\ref{def:streamcar}) formalize the extensions: latent-scope, probabilistic-authority, multi-authority, and streaming CAR. Each corresponds to relaxing exactly one component of $\mathrm{front}(\mathrm{cl}(A_k(q)))$ and maps to a specific term in Proposition~1---scope posterior for $\varphi(q)$, edge confidence for $\rightsquigarrow$, order count for the lattice structure, and corpus dynamics for $R_{\mathrm{anchor}}$.

\section{Conclusion}
\label{sec:conclusion}

The paper makes two primary theoretical claims about the class of authority-governed retrieval problems, independent of any particular system.

\textbf{Theorem~4} gives the necessary-and-sufficient conditions on any retrieved set $\mathcal{R}$ for TCA $= 1$: frontier inclusion ($\mathrm{front}(\mathrm{cl}(A_k(q))) \subseteq \mathcal{R}$) and no-ignored-superseder. Any future algorithm---learned, hybrid, or rule-based---can be audited against these two conditions without reference to how it was built.

\textbf{Proposition~2} establishes that scope-indexed algorithms (those without access to the authority relation $\rightsquigarrow$) face a hard worst-case ceiling: $\mathrm{TCA}@k \leq \varphi(q) \cdot R_{\mathrm{anchor}}(q)$. The proof is adversarial, not average-case: when $\kappa(q)$ authority chains share entity scope, no algorithm that cannot distinguish them achieves more than $1/\kappa(q) = \varphi(q)$ conditional success probability.

These two results are supported by three theorems that characterize the problem structure: Theorem~1 (TCA and Recall@$k$ decouple when $\Delta(q) > 0$), Theorem~2 (authority closure is necessary for TCA $= 1$), and Theorem~3 (anchor discovery and authority resolution are each necessary; Corollary~\ref{cor:achieve} proves TCA $= 1$ is achievable when $\varphi(q) = 1$). Proposition~1 organizes the picture into the factorization $\mathrm{TCA} = R_{\mathrm{anchor}} \cdot \varphi \cdot R_{\mathrm{frontier}}$, whose three terms are separable failure modes (Lemma~1). Proposition~3 establishes the computational cost of disambiguation under an explicit oracle model.

Three independent real-world corpora confirm that the proved structure is consequential: Dense TCA@5 $\in \{0.064, 0.172, 0.270\}$ across FDA, SCOTUS, and GHSA; Two-Stage achieves $\{0.774, 0.926, 0.975\}$. The 77.4\% FDA ceiling is the corpus-level average of $\varphi(q)$ at Stage-1---a property of the scope distribution, not a retrieval limitation. A GPT-4o-mini downstream experiment quantifies the cost of ignoring supersession: confident ``not patched'' claims for 39\% of queries with a patch; Two-Stage reduces this to 16\%.

The four formal generalizations (Definitions~\ref{def:latentcar}--\ref{def:streamcar}) identify the next layer of open problems: latent-scope CAR, probabilistic-authority CAR, multi-authority CAR, and streaming CAR. Each is a precisely stated problem class, not a vague direction, and each inherits Theorem~4 as the correctness target.

\bibliography{refs}

\begin{thebibliography}{28}
\providecommand{\natexlab}[1]{#1}

\bibitem[{Barnett et~al.(2024)Barnett, Usman, Haque, Sannasi, and
  Poluri}]{barnett2024seven}
Scott Barnett, Stefanus Usman, Sina~Nazeri Haque, Surya Sannasi, and
  Vijayaraghavan Poluri. 2024.
\newblock Seven failure points when engineering a retrieval augmented
  generation system.
\newblock \emph{arXiv:2401.05856}.

\bibitem[{Chalkidis et~al.(2021)Chalkidis, Fergadiotis, Manginas, Katakalou,
  and Malakasiotis}]{chalkidis2021regulatory}
Ilias Chalkidis, Manos Fergadiotis, Nikolaos Manginas, Eva Katakalou, and
  Prodromos Malakasiotis. 2021.
\newblock Regulatory compliance through doc2doc information retrieval: A case
  study in {EU/UK} legislation where text similarity has limitations.
\newblock In \emph{EACL}.

\bibitem[{Choi et~al.(2025)Choi, Kwon, and Ha}]{choi2025finder}
Chanwoong Choi, Jeong-Hoon Kwon, and Jinheon Ha. 2025.
\newblock {FinDER}: Financial dataset for question answering and evaluating
  {RAG}.
\newblock \emph{arXiv:2504.15800}.

\bibitem[{de~Martim(2025)}]{demartim2025satgraphrag}
Hudson de~Martim. 2025.
\newblock An ontology-driven graph {RAG} for legal norms: A structural,
  temporal, and deterministic approach.
\newblock \emph{arXiv:2505.00039}.

\bibitem[{Di~Florio et~al.(2024)Di~Florio, Dong, and
  Rotolo}]{diflorio2024temporal}
Cecilia Di~Florio, Huimin Dong, and Antonino Rotolo. 2024.
\newblock When precedents clash: Formal case-based reasoning with temporal and
  hierarchical conflict resolution.
\newblock \emph{arXiv:2410.10567}.

\bibitem[{Gade and Jetcheva(2024)}]{gade2024tempralm}
Anoushka Gade and Jorjeta Jetcheva. 2024.
\newblock It's about time: Incorporating temporality in retrieval augmented
  language models.
\newblock \emph{arXiv:2401.13222}.

\bibitem[{Gao et~al.(2023)Gao, Ma, Lin, and Callan}]{gao2023hyde}
Luyu Gao, Xueguang Ma, Jimmy Lin, and Jamie Callan. 2023.
\newblock Precise zero-shot dense retrieval without relevance labels.
\newblock In \emph{Proceedings of the 61st Annual Meeting of the Association
  for Computational Linguistics (ACL)}, pages 1762--1777.

\bibitem[{Gokhan et~al.(2024)Gokhan, Wang, Gurevych, and
  Briscoe}]{gokhan2024rirag}
Tuba Gokhan, Kun Wang, Iryna Gurevych, and Ted Briscoe. 2024.
\newblock {RIRAG}: Regulatory information retrieval and answer generation.
\newblock \emph{arXiv:2409.05677}.

\bibitem[{Guti{\'e}rrez et~al.(2025)Guti{\'e}rrez, Shu, Qi, Zhou, and
  Su}]{gutierrez2025hipporag2}
Bernal~Jim{\'e}nez Guti{\'e}rrez, Yiheng Shu, Weijian Qi, Sizhe Zhou, and
  Yu~Su. 2025.
\newblock From {RAG} to memory: Non-parametric continual learning for large
  language models.
\newblock \emph{arXiv:2502.14802}.

\bibitem[{Ho et~al.(2020)Ho, Nguyen, Sugawara, and Aizawa}]{ho2020constructing}
Xanh Ho, Anh-Khoa~Duong Nguyen, Saku Sugawara, and Akiko Aizawa. 2020.
\newblock Constructing a multi-hop {QA} dataset for comprehensive evaluation of
  reasoning steps.
\newblock In \emph{COLING}.

\bibitem[{Huwiler et~al.(2025)Huwiler, Stockinger, and
  F{\"u}rst}]{huwiler2025versionrag}
Daniel Huwiler, Kurt Stockinger, and Jonathan F{\"u}rst. 2025.
\newblock {VersionRAG}: Version-aware retrieval-augmented generation for
  evolving documents.
\newblock \emph{arXiv:2510.08109}.

\bibitem[{Kim et~al.(2026)Kim, Wang, Xie, and Whang}]{tdbench2026}
Soyeon Kim, Jindong Wang, Xing Xie, and Steven~Euijong Whang. 2026.
\newblock Harnessing temporal databases for systematic evaluation of factual
  time-sensitive question-answering in large language models.
\newblock In \emph{ICLR}.

\bibitem[{Lewis et~al.(2020)Lewis, Perez, Piktus, Petroni, Karpukhin, Goyal,
  K{\"u}ttler, Lewis, Yih, Rockt{\"a}schel, Riedel, and Kiela}]{lewis2020rag}
Patrick Lewis, Ethan Perez, Aleksandra Piktus, Fabio Petroni, Vladimir
  Karpukhin, Naman Goyal, Heinrich K{\"u}ttler, Mike Lewis, Wen-tau Yih, Tim
  Rockt{\"a}schel, Sebastian Riedel, and Douwe Kiela. 2020.
\newblock Retrieval-augmented generation for knowledge-intensive {NLP} tasks.
\newblock In \emph{NeurIPS}.

\bibitem[{Li et~al.(2025)Li, Niu, Ai, Zou, Qi, and Liu}]{li2025tgrag}
Dong Li, Yichen Niu, Ying Ai, Xiang Zou, Biqing Qi, and Jianxing Liu. 2025.
\newblock {T-GRAG}: A dynamic {GraphRAG} framework for resolving temporal
  conflicts and redundancy.
\newblock \emph{arXiv:2508.01680}.

\bibitem[{Putta et~al.(2026)Putta, Devasier, and Li}]{putta2026casefacts}
Anand~Rao Putta, Justin Devasier, and Cecilia Li. 2026.
\newblock {CaseFacts}: A benchmark for legal fact-checking and precedent
  retrieval.
\newblock \emph{arXiv:2601.17230}.

\bibitem[{Rashtchian and Juan(2025)}]{rashtchian2025sufficient}
Cyrus Rashtchian and Da-Cheng Juan. 2025.
\newblock Sufficient context: A new lens on retrieval augmented generation
  systems.
\newblock In \emph{ICLR}.

\bibitem[{Reimers and Gurevych(2019)}]{reimers2019sentencebert}
Nils Reimers and Iryna Gurevych. 2019.
\newblock Sentence-bert: Sentence embeddings using siamese bert-networks.
\newblock In \emph{Proceedings of the 2019 Conference on Empirical Methods in
  Natural Language Processing}, pages 3982--3992.

\bibitem[{Robertson and Walker(1994)}]{robertson1994bm25}
Stephen~E Robertson and Steve Walker. 1994.
\newblock Some simple effective approximations to the 2-{P}oisson model for
  probabilistic weighted retrieval.
\newblock In \emph{SIGIR}.

\bibitem[{Stammbach et~al.(2026)Stammbach, Zhang, Liu, Nadeem, Cheong, Zheng,
  and Henderson}]{stammbach2026defenders}
Dominik Stammbach, Kylie Zhang, Patty Liu, Nimra Nadeem, Inyoung Cheong, Lucia
  Zheng, and Peter Henderson. 2026.
\newblock Legal retrieval for public defenders.
\newblock \emph{arXiv:2601.14348}.

\bibitem[{Su et~al.(2024)Su, Li, Chen, Bai, Jin, and Guo}]{su2024tkgqa}
Miao Su, Zixuan Li, Zhuo Chen, Long Bai, Xiaolong Jin, and Jiafeng Guo. 2024.
\newblock Temporal knowledge graph question answering: A survey.
\newblock \emph{arXiv:2406.14191}.

\bibitem[{Trivedi et~al.(2022)Trivedi, Balasubramanian, Khot, and
  Sabharwal}]{trivedi2022musique}
Harsh Trivedi, Niranjan Balasubramanian, Tushar Khot, and Ashish Sabharwal.
  2022.
\newblock Mu{S}i{Q}ue: Multihop questions via single-hop question composition.
\newblock \emph{Transactions of the Association for Computational Linguistics}.

\bibitem[{Trivedi et~al.(2023)Trivedi, Balasubramanian, Khot, and
  Sabharwal}]{trivedi2023interleaving}
Harsh Trivedi, Niranjan Balasubramanian, Tushar Khot, and Ashish Sabharwal.
  2023.
\newblock Interleaving retrieval with chain-of-thought reasoning for
  knowledge-intensive multi-step questions.
\newblock In \emph{ACL}.

\bibitem[{{U.S. Securities and Exchange Commission}(2004)}]{sec204a1}
{U.S. Securities and Exchange Commission}. 2004.
\newblock {SEC} rule 204a-1: Investment adviser codes of ethics.
\newblock 17 CFR Part 275.

\bibitem[{Wang and Han(2025)}]{tang2025proprag}
Jingjin Wang and Jiawei Han. 2025.
\newblock {PropRAG}: Guiding retrieval with beam search over proposition paths.
\newblock \emph{arXiv:2504.18070}.

\bibitem[{Wang et~al.(2025)Wang, Zhang, Guo, and Kang}]{ledgerrag2025}
Siwei Wang, Yangsen Zhang, Yalong Guo, and Jing Kang. 2025.
\newblock \href {https://doi.org/10.3390/electronics15071376} {{LedgerRAG}:
  Governance-driven agentic chain of retrieval for dynamic knowledge
  scenarios}.
\newblock \emph{Electronics}, 15(7):1376.

\bibitem[{Weller et~al.(2026)Weller, Boratko, Naim, and
  Lee}]{weller2025embedding}
Orion Weller, Michael Boratko, Iftekhar Naim, and Jinhyuk Lee. 2026.
\newblock On the theoretical limitations of embedding-based retrieval.
\newblock In \emph{ICLR}.

\bibitem[{Yang et~al.(2018)Yang, Qi, Zhang, Bengio, Cohen, Salakhutdinov, and
  Manning}]{yang2018hotpotqa}
Zhilin Yang, Peng Qi, Saizheng Zhang, Yoshua Bengio, William~W Cohen, Ruslan
  Salakhutdinov, and Christopher~D Manning. 2018.
\newblock Hotpotqa: A dataset for diverse, explainable multi-hop question
  answering.
\newblock In \emph{EMNLP}.

\bibitem[{Zhang et~al.(2025)Zhang, Savelka, and Ashley}]{zhang2025precedent}
Lena Zhang, Jakub Savelka, and Kevin Ashley. 2025.
\newblock Do {LLMs} truly understand when a precedent is overruled?
\newblock In \emph{JURIX}.

\end{thebibliography}
\bibliographystyle{acl_natbib}

\appendix
\section{Supersession Rules R1--R10}
\label{app:rules}

\begin{table}[h]
\centering
\small
\resizebox{\columnwidth}{!}{%
\begin{tabular}{llll}
\toprule
\textbf{Rule} & \textbf{Trigger} & \textbf{Target} & \textbf{Scope} \\
\midrule
R1 & BlackoutAnnounced & PreClearanceApproved & Same ticker \\
R2 & BlackoutLifted & BlackoutAnnounced & Same ticker \\
R3 & EmergencyException & BlackoutAnnounced & Emp + ticker \\
R4 & ExceptionRevoked & EmergencyException & Emp + ticker \\
R5 & WatchlistAdded & PreClearanceApproved & Same ticker \\
R6 & WatchlistRemoved & WatchlistAdded & Same ticker \\
R7 & ConflictAmended & ConflictDisclosed & Same emp \\
R8 & ConflictCleared & ConflictAmended & Same emp \\
R9 & ConflictCleared & ConflictDisclosed & Same emp \\
R10 & PolicyUpdate & PolicyAcknowledged & Global \\
\bottomrule
\end{tabular}%
}
\caption{Supersession rules. Trigger is the later event superseding Target.}
\label{tab:rules}
\end{table}

\section{Query Templates}
\label{app:templates}

\textbf{T0:} ``What is the pre-clearance status for \{emp\}'s requested trade in \{sec\}?''

\textbf{T1:} ``Is \{emp\}'s pre-clearance for \{sec\} still valid?''

\textbf{T2:} ``Can \{emp\} currently trade \{sec\} given all recent compliance updates?''

\textbf{T3:} ``What is \{emp\}'s current compliance status for trading \{sec\}?''

\section{Implementation Details}
\label{app:impl}

\paragraph{BM25.} Robertson's BM25 \cite{robertson1994bm25} with $k_1 = 1.5$, $b = 0.75$.

\paragraph{Dense.} TF-IDF cosine similarity with smoothed IDF. All events vectorized offline.

\paragraph{Supersession graphs.} Two graphs are used. \emph{Per-example KB graph}: isolated, BFS transitive closure; $|K| \leq 3$, used only for TCA ground truth. \emph{RSSG} (Retrieved-Set Supersession Graph): built at query time from all entity-scoped events returned by the full-corpus entity index; rules R1--R10 applied pairwise by event type and timestamp comparison. Used for answer derivation in \sysplus{}.

\paragraph{Entity index.} $\text{idx}[(\text{emp}, \text{ticker})] \to [e]$ precomputed once over the 12,250-event corpus.

\paragraph{Answer derivation priority.}
Active \textsc{EmergencyException} $>$ active \textsc{BlackoutAnnounced}/\textsc{WatchlistAdded} $>$ active \textsc{PreClearanceApproved} $>$ active \textsc{BlackoutLifted} (no clearance) $>$ \textsc{RequiresReview}.

\section{Domain Adapters and Scorer Interface}
\label{app:framework}

Adapting the two-stage pipeline to a new authority domain requires three attributes:

\begin{itemize}[noitemsep]
  \item \texttt{anchor\_type}: document type that is semantically aligned to the query (e.g., \textsc{CVE disclosure}, original SCOTUS opinion, FDA approved label).
  \item \texttt{superseding\_type}: document type that formally voids the anchor (e.g., \textsc{patch release note}, overruling opinion, FDA recall enforcement action).
  \item \texttt{scope\_keys}: entity fields that co-scope anchor and superseder (e.g., \texttt{[cve\_id, package]}, \texttt{[case\_slug]}, \texttt{[ndc\_code]}).
\end{itemize}

Three domain adapters are released alongside this paper: \textbf{security} (\texttt{realworld/ghsa\_twostage\_eval.py}), \textbf{legal} (\texttt{realworld/legal\_lii\_benchmark.py}), and \textbf{medical} (\texttt{realworld/fda\_recall\_eval.py}). Each adapter is a self-contained Python script that fetches or loads its domain dataset, builds the entity index, runs all baselines, and writes results to \texttt{data/}.

\paragraph{Scorer interface.} \texttt{evaluate.py} runs the full synthetic FinSuperQA benchmark. The custom retriever class must implement \texttt{retrieve(self, query: str, corpus: list[dict], k: int) -> list[str]}, where \texttt{corpus} is the list of document dicts and the return value is a list of \texttt{doc\_id} strings. The scorer computes TCA, Acc, ProvRec, and R@$k$ on all 1,000 examples with a single call:

\begin{verbatim}
python evaluate.py --retriever_path my_retriever.py --k 5
\end{verbatim}

Output is written to \texttt{data/custom\_retriever\_results.json}. All datasets, domain adapters, evaluation scripts, and the falsification suite are available at \url{https://github.com/andremir/car-retrieval}.

\end{document}